\documentclass[11pt,a4paper]{article}

\usepackage[margin=1in]{geometry}
\usepackage[latin2]{inputenc}
\usepackage[english]{babel}
\usepackage[cmex10]{amsmath}
\usepackage{amsthm}
\usepackage{amssymb}
\usepackage{tabularx}
\usepackage{graphicx}
\usepackage{tikz}
\usepackage{xspace}
\usepackage{comment}
\usepackage{enumitem}
\usepackage{lineno}

\setlist{noitemsep}

\renewcommand{\H}{\mathcal{H}}

\newcommand{\X}{\mathcal{X}}

\newcommand{\executeiffilenewer}[3]{%
\ifnum\pdfstrcmp{\pdffilemoddate{#1}}%
{\pdffilemoddate{#2}}>0%
{\immediate\write18{#3}}\fi%
} 

\graphicspath{{figs/}}

\newcommand{\svg}[2]{\def\svgwidth{#1}%
\executeiffilenewer{figs/#2.svg}{figs/#2.pdf}%
{inkscape -z -D --file=figs/#2.svg %
--export-pdf=figs/#2.pdf --export-latex}%
{\input{figs/#2.pdf_tex}}}

\newtheorem{theorem}{Theorem}[section]
\newtheorem{lemma}[theorem]{Lemma}

\newtheorem{corollary}[theorem]{Corollary}
\newtheorem{proposition}[theorem]{Proposition}
\theoremstyle{definition}
\newtheorem{definition}[theorem]{Definition}
\newtheorem{remark}[theorem]{Remark}

\newcommand{\ba}{\mathbf{a}}
\newcommand{\bb}{\mathbf{b}}
\newcommand{\bx}{\mathbf{x}}
\newcommand{\be}{\mathbf{e}}
\newcommand{\tw}{\textup{tw}}
\newcommand{\csp}{\textup{CSP}}

\newcommand{\G}{\mathcal{G}}
\newcommand{\R}{\mathcal{R}}
\newcommand{\dirG}{\overrightarrow{G}}

\newcommand{\mya}{\mathbf{a}}
\newcommand{\myb}{\mathbf{b}}
\newcommand{\xx}{\mathbf{x}}
\newcommand{\yy}{\mathbf{y}}
\newcommand{\zz}{\mathbf{z}}
\newcommand{\uu}{\mathbf{u}}
\newcommand{\vv}{\mathbf{v}}
\newcommand{\ww}{\mathbf{w}}
\newcommand{\maxflow}{\mathsf{max\text{-}flow}}
\newcommand{\dmax}{\Delta}

\newcommand{\dem}{\mathsf{\mathsf{d}}}
\newcommand{\capa}{\mathsf{\mathsf{c}}}

\newcommand{\Rs}{\textup{\sf R}}
\newcommand{\Hs}{\textup{\sf H}}

\newif\ifappendix
\newif\ifmain
\newif\ifmove
\maintrue
\movefalse
\newif\ifabstract
\newif\iffull

\ifabstract \fullfalse \else \fulltrue \movetrue \fi 

\ifabstract
\usepackage{mathptmx}
\newcommand{\appstar}{$[*]$\ }
\usepackage[nomarkers,nolists]{endfloat}

\linenumbers
\DeclareMathAlphabet{\mathcal}{OMS}{cmsy}{m}{n}
\else
\newcommand{\appstar}{}
\fi

\renewcommand{\phi}{\varphi}
\newcommand{\eps}{\varepsilon}

\newcommand{\dommax}{\delta}

\usepackage[margin=1in]{geometry}

\begin{document}
\title{The limited blessing of low dimensionality: when $1-1/d$ is the best possible exponent for $d$-dimensional geometric problems}
\author{
D\'aniel Marx\thanks{Institute of Computer Science and Control,
 Hungarian Academy of Sciences (MTA SZTAKI), Budapest, Hungary.  Research supported by the European Research Council (ERC)  grant 
``PARAMTIGHT: Parameterized complexity and the search for tight
complexity results,'' reference 280152 and OTKA grant NK105645. \texttt{dmarx@cs.bme.hu}}
\and
Anastasios Sidiropoulos\thanks{Dept.~of Computer Science \& Engineering, and Dept.~of Mathematics,
The Ohio State University, Columbus, OH, USA. \texttt{sidiropoulos.1@osu.edu}}
%{\tt dmarx@cs.bme.hu}
}
\date{}

\maketitle
\ifabstract
\thispagestyle{empty}
\setcounter{page}{0}  
\fi
\begin{abstract}
  We are studying $d$-dimensional geometric problems that have
  algorithms with $1-1/d$ appearing in the exponent of the running
  time, for example, in the form of $2^{n^{1-1/d}}$ or
  $n^{k^{1-1/d}}$. This means that these algorithms perform somewhat
  better in low dimensions, but the running time is almost the same
  for all large values $d$ of the dimension. Our main result is
  showing that for some of these problems the dependence on $1-1/d$ is
  best possible under a standard complexity assumption. We show that,
  assuming the Exponential Time Hypothesis,
\begin{itemize}
\item $d$-dimensional Euclidean TSP on $n$ points cannot be solved in
  time $2^{O(n^{1-1/d-\epsilon})}$ for any $\epsilon>0$, and
\item the problem of finding a set of $k$ pairwise nonintersecting
  $d$-dimensional unit balls/axis parallel unit cubes cannot be solved
  in time $f(k)n^{o(k^{1-1/d})}$ for any computable function $f$.
\end{itemize}
These lower bounds essentially match the known algorithms for these
problems. To obtain these results, we first prove lower bounds on the
complexity of Constraint Satisfaction Problems (CSPs) whose constraint
graphs are $d$-dimensional grids. We state the complexity results on
CSPs in a way to make them convenient starting points for
problem-specific reductions to particular $d$-dimensional geometric
problems and to be reusable in the future for further results of
similar flavor.
\end{abstract}

\ifabstract
\clearpage
\fi
\section{Introduction}

The curse of dimensionality is a ubiquitous phenomenon observed over
and over again for geometric problems: polynomial\-/time algorithms that
work for low dimensions quickly become infeasible in high dimensions,
as the running time depends exponentially on the dimension
$d$. Consider, for example, the Euclidean $k$-center problem, which
can be formulated as follows: given a set of $P$ points in
$d$-dimensions, find a set of $k$ unit balls whose union covers
$P$. For $k=1$, the problem can be solved in linear time
\cite{DBLP:journals/siamcomp/Megiddo83a}, but the problem becomes
NP-hard even for $k=2$ \cite{DBLP:journals/jsc/Megiddo90} and only
algorithms with running time of the form $n^{O(d)}$ is known \cite{DBLP:journals/csur/AgarwalS98}.  In
recent years, the framework of W[1]-hardness has been used to give
evidence that for several problems (including Euclidean 2-center), the
exponent of $n$ has to depend on the dimension $d$; in fact, for many
of these problems tight lower bounds have been given that show that no
$n^{o(d)}$ algorithm is possible under standard complexity assumptions
\cite{DBLP:journals/comgeo/GiannopoulosKRW13,DBLP:journals/jc/GiannopoulosKWW12,DBLP:journals/talg/CabelloGKMR11,DBLP:conf/stacs/KnauerTW11,DBLP:conf/tamc/Knauer10,DBLP:conf/iwpec/GiannopoulosKR09,DBLP:journals/ipl/CabelloGK08,DBLP:conf/compgeom/Chan08b}.

For certain other geometric problems, however, the dimension $d$ affects the
complexity of the problem in a very different way. Consider, for
example, the classical Traveling Salesperson Problem (\textsc{TSP}): given a
distance metric on a set of $n$ points, the task is to find a shortest
path\footnote{This variant of TSP is also known as \emph{path-TSP}.
Our lower bound also holds for \emph{tour-TSP}, i.e.~the variant where one seeks to find a cycle visiting all vertices. In order to simplify the discussion, we restrict our attention to path-TSP for now; on Section \ref{sec:TSP} we explain how our proof can be modified to obtain the same lower bound for cycle-TSP.} visiting all $n$ points. This problem can be solved in time
$2^n\cdot n^{O(1)}$ using a well-known dynamic programming algorithm
of Bellman \cite{Bellman:1962:DPT:321105.321111} and of Held and Karp \cite{heldkarp} that works for any
metric. However, in the special case when the points are in
$d$-dimensional Euclidean space, \textsc{TSP} can be solved in time
$2^{d^{O(d)}}\cdot n^{O(dn^{1-1/d})}$ by an algorithm of Smith and Wormald~\cite{DBLP:conf/focs/SmithW98},
that is, treating $d$ as a fixed constant, the running time is
$n^{O(n^{1-1/d})}=2^{O(n^{1-1/d}\cdot \log
  n)}=2^{O(n^{1-1/d+\epsilon})}$ for every $\epsilon>0$. This means
that, as the dimension $d$ grows, the running time quickly converges
to the $2^n\cdot n^{O(1)}$ time of the standard dynamic programming
algorithm that does not exploit any geometric property of the
problem. On the other hand, when the dimension $d$ is small, the
algorithm has a moderate gain over dynamic programming: for example,
for $d=2$, we have $2^{O(\sqrt{n}\log n)}$ instead of $2^{n}\cdot
n^{O(1)}$. This behavior is very different compared to the $n^{O(d)}$
running time of algorithms for problems affected by the curse of
dimensionality: for those problems, complexity gets constantly worse and worse as $d$ grows, while for \textsc{TSP} the complexity is essentially the same
for all large values of $d$. Therefore, we may call this phenomenon observed for $d$-dimensional \textsc{TSP} 
the ``limited blessing of low dimensionality'': the running time is
almost uniformly bad for large values of $d$, but some amount of
improvement can be achieved for low dimensions. 

A slightly different example of the same phenomenon appears in the
case of packing problems.  Consider the following problem: Given a set
of $n$ unit balls in $d$-dimensional space and an integer $k$, the
task is to find $k$ pairwise disjoint balls, or in other words, we
have to find an independent set of size $k$ in the intersection graph
of the balls. Clearly, this can be done in time $n^{O(k)}$ by brute
force for any intersection graph. However, using the geometric nature
of the problem, one can reduce the running time to
$n^{O(k^{1-1/d})}$ (for $d=2$, this has been proved by Alber and Fiala \cite{MR2070519}; in Appendix~\ref{sec:packing-alg}, we sketch a simple algorithm for any $d\ge 2$  based on a standard sweeping argument and dynamic programming\footnote{We thank Sariel Har-Peled for suggesting this approach for the algorithm.}).  Again, we are in a similar situation as in the
case of $d$-dimensional \textsc{TSP}: as $d$ grows, the running time quickly
converges to the $n^{O(k)}$ running time of brute force, but there is
a moderate improvement for low dimensions (for example,
$n^{O(\sqrt{k})}$ vs.~$n^{O(k)}$ for $d=2$).

\textbf{Our results.}  Can we make the blessing of low dimensionality
more pronounced? That is, can we improve $1-1/d$ in the exponent of
the running time of the algorithms described above to something like
$1-1.1/d$ or $1-1/\sqrt{d}$? The main result of the current paper is
showing that the exponent $1-1/d$ is best possible for these problems.
We prove these results under the complexity assumption called {\em
  Exponential Time Hypothesis (ETH)}, introduced by Impagliazzo,
Paturi, and Zane \cite{MR1894519}, stating that $n$-variable 3SAT
cannot be solved in time $2^{o(n)}$. This complexity assumption is the
basis of tight lower bounds for many problems, see the survey of
Lokshtanov et al.~\cite{survey-eth-beatcs}.

For $d$-dimensional \textsc{TSP}, we prove the following result:
\begin{theorem}\label{th:intro-tsp}
If for some $d\ge 2$ and $\epsilon>0$, \textsc{TSP} in $d$-dimensional Euclidean space can be solved in time $2^{O(n^{1-1/d-\epsilon})}$, then ETH fails.
\end{theorem}
Note that this lower bound essentially matches the
$2^{O(n^{1-1/d+\epsilon})}$ time algorithm of Smith and
Wormald~\cite{DBLP:conf/focs/SmithW98}. For packing problems, we
have the following results:
\begin{theorem}\label{th:intro-packing}
  If for some $d\ge 2$ and computable function $f$, there is a
  $f(k)n^{o(k^{1-1/d})}$ time algorithm for finding $k$ pairwise
  nonintersecting $d$-dimensional balls/axis-parallel cubes, then ETH
  fails.
\end{theorem}
That is, the exponent $k^{1-1/d}$ cannot be improved, even if we allow
an arbitrary dependence $f(k)$ as a multiplicative constant. That is, in the
language of parameterized complexity, we are not only proving that the
problem is not fixed-parameter tractable, but we also give a tight
lower bound on the dependence of the exponent on the parameter $k$.
%   Note
% that this implies the same form of lower bound as in
% Theorem~\ref{th:intro-tsp} (no $2^{O(n^{1-1/d-\epsilon})}$ time
%   algorithm), but this result is stronger, as it holds even if $k$ is
%   restricted to be much smaller than $n$.

To prove Theorems~\ref{th:intro-tsp} and \ref{th:intro-packing}, we
first develop general tools for approaching $d$-dimensional geometric
problems. We formulate complexity results in the abstract setting of
Constraint Satisfaction Problems (CSPs) whose constraint graphs are
$d$-dimensional grids. These results faithfully capture the influence
of the number $d$ of dimensions on problem complexity and are stated
in a way to facilitate further reductions to $d$-dimensional geometric
problems. Then we can obtain Theorems~\ref{th:intro-tsp} and
\ref{th:intro-packing} by problem-specific reductions that are mostly
local and do not depend very much on the number of dimensions. We
believe that our results for $d$-dimensional CSPs could serve as a
useful starting point for proving further results of this flavor for
geometric problems.  Producing an exhaustive list of such results was
not the goal of the current paper; instead, we wanted to demonstrate
that $1-1/d$ in the exponent can be the best possible dependence on
the dimension, build the framework for proving such lower bounds, and
provide a sample of results on concrete problems.

Let us remark that the results in
Theorems~\ref{th:intro-tsp}--\ref{th:intro-packing} were already known
for the special case of $d=2$. Papadimitriou
\cite{DBLP:journals/tcs/Papadimitriou77} proved the NP-hardness of
Euclidean \textsc{TSP} in $d=2$ dimensions by a reduction from
\textsc{Exact-Cover}: given an instance of \textsc{Exact-Cover} with
universe size $n$ and $m$ subsets, the reduction creates an equivalent
instance of \textsc{TSP} with $O(nm)$ points. It can be shown that an
instance of \textsc{3-Coloring} with $n$ vertices and $m$ edges can be reduced
to an instance of \textsc{Exact-Cover} with universe size $O(n+m)$ and
number of sets $O(n+m)$. Therefore, a $2^{o(\sqrt{n})}$ algorithm for
TSP in $d=2$ dimensions would give a $2^{o(n+m)}$ time algorithm for
\textsc{3-Coloring}, contradicting ETH \cite{survey-eth-beatcs}.

Marx~\cite{marx-ptaslower} proved the W[1]-hardness of finding $k$
pairwise nonintersecting unit disks (in $d=2$ dimensions) by a
reduction from \textsc{$k$-Clique}. The reduction maps an instance of
\textsc{$k$-Clique} to a set of disks where $k':=k^2$ independent disks have to
be found. By a result of Chen et al.~\cite{MR2121603}, if \textsc{$k$-Clique}
can be solved in time $f(k)n^{o(k)}$ for some computable function $f$,
then ETH fails. Putting together the result of Chen et
al.~\cite{MR2121603} and the reduction of Marx~\cite{marx-ptaslower},
we get Theorem~\ref{th:intro-packing} for $d=2$.

For $d\ge 3$ dimensions, however, the tight lower bounds become much
harder to obtain. As we shall see, the hardness proofs rely on
constructing embeddings into $d$-dimensional grids. For $d=2$, this
can be achieved by simple and elementary arguments, but the tight
results for $d\ge 3$ require more delicate constructions.

\textbf{Constraint satisfaction problems.} We use the language of CSPs
to express the basic lower bounds in a way that is not specific to any
geometric problem. A CSP is defined by a set $V$ of variables, a
domain $D$ from which the variables can take values, and a set of
constraints on the variables. In the current paper, we consider only
CSPs where every constraint is binary, that is, involves only two
variables and restricts the possible combination of values that can
appear on those two variables in a solution (see
Section~\ref{sec:constr-satisf-probl} for definitions related to
CSPs). It is important to point out that one can consider CSPs where
the size of the domain $D$ is a fixed small constant (e.g., \textsc{3-Coloring}
on a graph $G$ can be reduced to a CSP with $|V(G)|$ variables and $|D|=3$) or CSPs where the domain size is
large, much larger than the number of variables (e.g., \textsc{$k$-Clique} on a
graph $G$ can be reduced to a CSP with $k$ variables and $|D|=|V(G)|$). We
will need both viewpoints in the current paper.

Intuitively, it is clear how a hardness proof for a $d$\-/dimensional
geometric problem should proceed. We construct small gadgets able to represent a certain number of states and put copies of these gadgets at certain
locations in $d$-dimensional space. Then each gadget can interact with
the at most $2d$ gadgets that are ``adjacent'' to it in one of the $d$
dimensions. The gadgets should be constructed to ensure that each such
interaction enforces some binary constraint on the states of the two
gadgets. Therefore, we can effectively express a CSP where the
variables are located on the $d$-dimensional grid and the binary
constraints are only on adjacent variables. This means that we need lower
bounds on the complexity of such CSPs. In particular, we would like to
understand the complexity of CSPs where the graph of constraints is exactly 
a $d$-dimensional grid.

There is a large body of literature on how {\em structural
  restrictions,} that is, restrictions on the constraint graph
influences the complexity of CSP
\cite{chen-grohe-succinct,1206036,MR2351517,grohe-marx,MR2213853,380867,DBLP:journals/toc/Marx10,marx-fractional-talg,marx-tocs-truthtable,gotleosca02}. Specifically,
we need a general result of Marx~\cite{DBLP:journals/toc/Marx10}
stating that, in a precise technical sense, treewidth of the
constraint graph governs the complexity of the problem. Roughly
speaking, the result states that, assuming the Exponential Time
Hypothesis, there is no $|D|^{o(\tw(G)/\log \tw(G))}$ time algorithm
for CSP (where $\tw(G)$ is the treewidth of the constraint graph) and
this holds {\em even if we restrict the constraint graph to any class
  of graphs.}  Therefore, if we restrict CSP to instances where the
constraint graph is the $d$-dimensional grid with $k=m^d$ vertices for
some $m$, then the known fact that such a $d$-dimensional grid has
treewidth $O(m^{d-1})=O(k^{1-1/d})$ implies that there is no
$|D|^{o(k^{1-1/d}/\log k)}$ time algorithm for such CSPs.  Therefore,
in a sense, the connection to treewidth given by
\cite{DBLP:journals/toc/Marx10} and the treewidth of the
$d$-dimensional grid explain why $1-1/d$ is the right exponent for the
$d$-dimensional geometric problems we are considering.

We still have some work left to prove
Theorems~\ref{th:intro-tsp}--\ref{th:intro-packing}. First, as a minor
issue, we remove the log factor from the lower bound obtained above
for CSPs whose constraint graph is a $d$-dimensional grid and improve
it to the tight bound ruling out $|D|^{o(k^{1-1/d})}$ time algorithms. The general result of
\cite{DBLP:journals/toc/Marx10} is a based on constructing certain
embeddings exploiting the treewidth of the constraint graph. However, by
focusing on a specific class of graphs, we can obtain slightly better
embeddings and therefore improve the lower bound. In particular, Alon and Marx
\cite{DBLP:journals/siamdm/AlonM11} gave an embedding of an arbitrary graph into a
$d$-dimensional Hamming graph (generalized hypercube) and it is
easy to embed a $(d-1)$-dimensional Hamming graph into a
$d$-dimensional grid. These embeddings together prove the tighter
lower bound. Second, we modify the CSPs to make them more suited to
reductions to geometric problems. In these CSPs, the domain is
$[\dommax]^d$ for some integer $\dommax$, that is, the solution assigns every variable a $d$-tuple of integers as a value. Every constraint is of the same form:
if variables $v_1$ and $v_2$ are adjacent in the $i$-th dimension (with $v_2$ being larger by one in the $i$-th coordinate),
then the constraint requires that the $i$-th component of the value of
$v_1$ is at most the $i$-th component of the value of $v_2$.  Then
problem-specific, but very transparent reductions from these CSPs to packing unit
disks or unit cubes prove Theorem~\ref{th:intro-packing}.

To prove Theorem~\ref{th:intro-tsp}, we need a slightly different
approach. The issue is that the general result of
\cite{DBLP:journals/toc/Marx10} holds only if there is no bound on the
domain size. Thus we need a reduction to \textsc{TSP} that works even
if the domain size is much larger than the number of
variables. However, if we have $k$ variables and domain size $|D|$,
then probably the best we can hope for is a reduction to \textsc{TSP}
with $n=O(|D|k)$ or so points (and even this is only under the
assumption that we can construct gadgets with $O(|D|)$ points to
represent each variable and each constraint, which is far from
obvious). But then a $2^{O(n^{1-1/d-\epsilon})}$ time algorithm for
$d$-dimensional \textsc{TSP} would give only a
$2^{O(|D|k)^{1-1/d-\epsilon}}$ time algorithm for CSP, which would not
violate the lower bound ruling out $|D|^{o(k^{1-1/d})}$ time
algorithms. Therefore, we prove a variant of the lower bound stating
that there is a constant $\dommax$ such that there is no
$2^{O(k^{1-1/d-\epsilon})}$ time algorithm for CSP on $d$-dimensional
grids even under the restriction $|D|\le \dommax$. Again, we prove this
lower bound by revisiting the embedding results into $d$-dimensional
grids and Hamming graphs. Then a problem-specific reduction reusing some of the ideas of
Papadimitriou~\cite{DBLP:journals/tcs/Papadimitriou77} 
 for
the $d=2$ case 
proves
Theorem~\ref{th:intro-tsp}. Interestingly, our reduction exploits the
fact $d\ge 3$: this allows us to express arbitrary binary relations in
an easy way without having to worry about crossings.

\iffull
\textbf{Organization.} The paper is organized as follows. Section~\ref{sec:constr-satisf-probl} introduces CSPs and proves the lower bounds on CSPs  whose constraint graphs are $d$-dimensional grids. Section~\ref{sec:packing} transfers these lower bounds to packing unit balls/cubes. Section~\ref{sec:TSP} proves the result on TSP. Appendix~\ref{sec:embedding_alternative} gives an alternative proofs for one of the embedding results of Section~\ref{sec:constr-satisf-probl}. Appendix~\ref{sec:packing-alg} presents an $n^{O(k^{1-1/d})}$ time algorithm for packing unit balls/cubes in $d$-dimensions.
\fi

%%% Local Variables: 
%%% mode: latex
%%% TeX-master: "main"
%%% End: 

%\input{prelim}
\ifmain
\ifabstract
\section{Constraint Satisfaction\\ Problems}
\else
\section{Constraint Satisfaction Problems}
\fi
\label{sec:constr-satisf-probl}

Understanding constraint satisfaction problems (CSPs) whose constraint
graphs are $d$-dimensional grids seems to be a very convenient starting
point for proving lower bounds on the complexity of $d$-dimensional
geometric problems. In this section, we review the relevant background
on CSPs and prove the basic complexity results that will be useful for
the lower bounds on specific $d$-dimensional geometric problems.

\begin{definition}\label{def:csp}
An instance of a {\em constraint satisfaction problem} is a triple $(V ,D, C)$,
where:
\begin{itemize}
\item $V$ is a set of variables,
\item $D$ is a domain of values,
\item $C$ is a set of constraints, $\{c_1,c_2,\dots ,c_q\}$.
Each constraint $c_i\in C$ is a pair $\langle
s_i,R_i\rangle$, where:
\begin{itemize}
\item $s_i$ is a tuple of variables of length $m_i$, called the {\em constraint scope,} and
\item $R_i$ is an $m_i$-ary relation over $D$, called the {\em constraint
  relation.}
\end{itemize}
\end{itemize}
\end{definition}
For each constraint $\langle s_i,R_i\rangle$ the tuples of $R_i$
indicate the allowed combinations of simultaneous values for the
variables in $s_i$. The length $m_i$ of the tuple $s_i$ is called the
{\em arity} of the constraint.  A {\em solution} to a constraint
satisfaction problem instance is a function $f$ from the set of
variables $V$ to the domain of values $D$ such that for each
constraint $\langle s_i,R_i\rangle$ with $s_i = (
v_{i_1},v_{i_2},\dots,v_{i_m})$, the tuple $(f(v_{i_1}),
f(v_{i_2}),\dots,f(v_{i_m}))$ is a member of $R_i$.  We say that
an instance is {\em binary} if each constraint relation is binary,
i.e., $m_i=2$ for each constraint (hence the term ``binary'' refers to the arity of the constraints and {\em not} to the size of the domain).
Note that Definition~\ref{def:csp} allows for a variable to appear multiple times in the scope of the constraint. Thus a binary instance can contain a constraint 
of the form $\langle (v,v),R\rangle$, which is essentially a unary constraint.
We will deal only with binary CSPs in this paper. We may assume that there is at most one constraint with the same scope, as two constraints $\langle s,R_1\rangle$ 
and $\langle s,R_2\rangle$ can be merged into a single constraint $\langle s,R_1\cap R_2\rangle$. Therefore, we may assume that the input size $|I|$ of a binary CSP instance is polynomial in $|V|$ and $|D|$, without going into the details of how the constraints are exactly represented.

The {\em primal graph} (or {\em Gaifman graph}) of a CSP instance
$I=(V ,D, C)$ is a graph with vertex set $V$ such that distinct
vertices $u,v\in V$ are adjacent if and only if there is a constraint
whose scope contains both $u$ and $v$. The following classical result
shows that treewidth of the primal graph is a relevant parameter to
the complexity of CSPs: low treewidth implies efficient algorithms.
\begin{theorem}[Freuder~\cite{Freuder90AA}]\label{th:freuder}
Given a binary CSP instance $I$ whose primal graph has treewidth $w$, a solution can be found in time $|I|^{O(w)}$.
\end{theorem}

The {\em $d$-dimensional grid} $\Rs[n,d]$ has vertex set $[n]^d$ and
vertices $\ba=(a_1,\dots, a_d)$ and $\bb=(b_1,\dots, b_d)$ are
adjacent if and only if $\sum_{i=1}^d|a_i-b_i|=1$, that is, they
differ in exactly one coordinate and only by exactly one in that
coordinate. In other words, if we denote by $\be_i$ the unit vector
whose $i$-th coordinate is 1 and every other coordinate is 0, then
$\bb$ is the neighbor of $\ba$ only if $\bb$ is of the form $\ba+\be_i$
or $\ba-\be_i$ for some $1\le i \le d$. Note that the maximum degree
of $\Rs[n,d]$ is $2d$ (for $n\ge 3$). We denote by $\R_d$ the set of
graphs $\Rs[n,d]$ for every $n\ge 1$. For every fixed $d$, the treewidth
of the $d$-dimensional grid is $\Theta(n^{d-1})$ (this is proved for
the related notion of carving width by Kozawa et
al.~\cite{DBLP:journals/dm/KozawaOY10}, but carving width is known to
be between $\tw(G)/3$ and $\Delta\tw(G)$, where $\Delta$ is the maximum degree
\cite{DBLP:conf/isaac/ThilikosSB00}).

\ifabstract
\begin{proposition}[Kozawa et
al.~\cite{DBLP:journals/dm/KozawaOY10}]\label{prop:gridtw}
For any fixed $d\ge 2$, the treewidth of $\Rs[n,d]$ is $\Theta(n^{d-1})=\Theta(|V(\Rs[n,d])|^{1-1/d})$.
\end{proposition}
\else
\begin{proposition}[Kozawa et
al.~\cite{DBLP:journals/dm/KozawaOY10}]\label{prop:gridtw}
For every fixed $d\ge 2$, the treewidth of $\Rs[n,d]$ is $\Theta(n^{d-1})=\Theta(|V(\Rs[n,d])|^{1-1/d})$.
\end{proposition}
\fi

Theorem~\ref{th:freuder} and Proposition~\ref{prop:gridtw} together
imply that, for every fixed $d$, CSPs restricted to instances where
the primal graph is a $d$-dimensional grid $\Rs[n,d]$ can be solved in time
$|I|^{O(n^{d-1})}=|I|^{O(|V|^{1-1/d})}$.

A result of Marx~\cite{DBLP:journals/toc/Marx10} provides a converse
of Theorem~\ref{th:freuder} showing, in a precise technical sense,
that it is indeed the treewidth of the primal graph that determines
the complexity.  This very general result can be used to provide an
almost matching lower bound on the complexity of solving CSPs on a
$d$-dimensional grid. For a class $\G$ of graphs, let us denote by
$\csp(\G)$ the binary CSP problem restricted to instances whose primal
graph is in $\G$.
\begin{theorem}[Marx~\cite{DBLP:journals/toc/Marx10}]\label{th:binarymain}
  If there is a recursively enumerable class $\G$ of graphs with unbounded treewidth, an
  algorithm $\mathbb{A}$, and a function $f$ such that $\mathbb{A}$
  correctly decides every binary $\csp(\G)$ instance in time
  $f(|V|)|I|^{o({\tw}(G)/\log {\tw}(G))}$, then ETH fails.
\end{theorem}
Theorem~\ref{th:binarymain} and Proposition~\ref{prop:gridtw} together imply the following lower bound:
\ifabstract
\begin{corollary}\label{cor:grid}
For every fixed $d\ge 2$, there is no\\ $f(|V|)|I|^{o(|V|^{1-1/d}/\log|V|)}$ algorithm for $\csp(\R_d)$ for any function $f$, unless ETH fails.
\end{corollary}
\else
\begin{corollary}\label{cor:grid}
For every fixed $d\ge 2$, there is no $f(|V|)|I|^{o(|V|^{1-1/d}/\log|V|)}$ algorithm for $\csp(\R_d)$ for any function $f$, unless ETH fails.
\end{corollary}
\fi
We extend Corollary~\ref{cor:grid} in two ways. First, by an analysis specific to the class $\R_d$ of graphs at hand and avoiding the general tools used in the proof of Theorem~\ref{th:binarymain}, we can get rid of the logarithmic factor in the exponent, making the result tight.
\ifabstract
\begin{theorem}\appstar\label{th:gridnolog}
For every fixed $d\ge 2$, there is no\\ $f(|V|)|I|^{o(|V|^{1-1/d})}$ algorithm for $\csp(\R_d)$ for any function $f$, unless ETH fails.
\end{theorem}
\else
\begin{theorem}\label{th:gridnolog}
For every fixed $d\ge 2$, there is no $f(|V|)|I|^{o(|V|^{1-1/d})}$ algorithm for $\csp(\R_d)$ for any function $f$, unless ETH fails.
\end{theorem}
\fi
 Second, observe that Theorem~\ref{th:binarymain} does not give any lower bound for the restriction of the problem to instances with domain size bounded by a fixed constant $\dommax$. In fact, no such strong negative result as Theorem~\ref{th:binarymain} can hold for instances with domain size restricted to $\dommax$: as the size of the instance is bounded by a function of $|V|$ and $\dommax$, it can be solved in time $f(|V|)$ for some function $f$ (assuming $\dommax$ is a constant).  Therefore, we prove a bound of the following form:
\begin{theorem}\appstar\label{th:gridbounded-domain}
For every fixed $d\ge 2$, there is a constant $\dommax_d$ such that there is no $2^{O(|V|^{1-1/d-\epsilon})}$ algorithm for $\csp(\R_d)$ with domain size at most $\dommax_d$ for any $\epsilon>0$, unless ETH fails.
\end{theorem}

%We prove Theorems~\ref{th:gridnolog} and \ref{th:gridbounded-domain} in {\iffull Sections~\ref{sec:proofnolog} and \ref{sec:proof-domain},\else Section~\ref{sec:proofnolog} and Appendix~\ref{sec:proof-domain},\fi} respectively.
{\iffull 
We prove Theorems~\ref{th:gridnolog} and \ref{th:gridbounded-domain} in Sections~\ref{sec:proofnolog} and \ref{sec:proof-domain}, respectively.\fi}
For the proof of Theorem~\ref{th:gridnolog} we show how
$d$\-/dimensional Hamming graphs can be embedded into
$(d+1)$-dimensional grids and then invoke a lower bound on
$d$-dimensional Hamming graphs by Alon and Marx
\cite{DBLP:journals/siamdm/AlonM11}. To prove
Theorem~\ref{th:gridbounded-domain}, we need to tighten the lower
bound on $d$-dimensional Hamming graphs by revisiting the embedding
result of Alon and
Marx \cite{DBLP:journals/siamdm/AlonM11}. Interestingly, this can be
done in two ways: either by a construction similar to the one by Alon
and Marx~\cite{DBLP:journals/siamdm/AlonM11} together with the randomized rounding of Raghavan and Thompson~\cite{DBLP:journals/combinatorica/RaghavanT87} or by deducing it from
the connection between multi-commodity flows and graph Laplacians.
\ifabstract
Due to lack of space, the proofs of Theorems~\ref{th:gridnolog}
 and \ref{th:gridbounded-domain} are
%is
 given in the full version of this paper.
\fi

%\fi
%\iffull
%\subsection{Proof Theorem~\ref{th:gridnolog}}
%\label{sec:proofnolog}
%\fi
%\ifmove
%\ifappendix\section{Proof Theorem~\ref{th:gridnolog}}
%\label{sec:proofnolog}
%\fi

\iffull
\subsection{Proof of Theorem~\ref{th:gridnolog}}
\label{sec:proofnolog}
The {\em $d$-dimensional Hamming graph} $\Hs[n,d]$ has vertex set
$[n]^d$, and vertices $(a_1,\dots, a_d)$ and $(b_1,\dots, b_d)$ are
adjacent if and only if there is exactly one $1\le i \le d$ with
$a_i\neq b_i$, that is, the Hamming distance of the two $d$-dimensional vectors is
1. We denote by $\H_d$ the set of graphs $\Hs[n,d]$ for every $n\ge 1$. For
fixed $d$, the treewidth of $\Hs[n,d]$ is $\Theta(n^d)$
\cite{SunilChandran2006359}. Therefore, the general result of
Theorem~\ref{th:binarymain} implies that there is no
$f(|V|)|I|^{o(n^d/\log n)}=f(|V|)|I|^{o(|V|/\log|V|)}$ algorithm for
$\csp(\H_d)$. A tight lower bound was given for this specific family
of graphs by Alon and Marx \cite{DBLP:journals/siamdm/AlonM11}:
\ifabstract
\begin{theorem}[\cite{DBLP:journals/siamdm/AlonM11}]\label{th:hamminglowerbound}
For every fixed $d\ge 2$, there is no\\ $f(|V|)|I|^{o(|V|)}$ algorithm for $\csp(\H_d)$ for any function $f$, unless ETH fails.
\end{theorem}
\else
\begin{theorem}[\cite{DBLP:journals/siamdm/AlonM11}]\label{th:hamminglowerbound}
For every fixed $d\ge 2$, there is no $f(|V|)|I|^{o(|V|)}$ algorithm for $\csp(\H_d)$ for any function $f$, unless ETH fails.
\end{theorem}
\fi
As we will need the same line of arguments for the proof of Theorem~\ref{th:gridbounded-domain} in
\iffull Section~\else Appendix~\fi\ref{sec:proof-domain}, let us review how one can prove
results such as Theorem~\ref{th:hamminglowerbound}
\cite{DBLP:journals/siamdm/AlonM11,DBLP:journals/toc/Marx10}.  We need
the following notion of embedding to transfer hardness results of CSP
problems on different graphs.  An {\em embedding} of a graph $F$ into a graph $G$ is a
mapping $\phi:V(F)\to 2^{V(G)}$ satisfying the following two
properties:
\begin{enumerate}
\item  For every vertex $u\in V(F)$, the set $\phi(u)$ induces a connected subgraph of $G$, and
\item  For every edge $uv\in E(F)$, the sets $\phi(u)$ and $\phi(v)$ {\em touch,} that is, either they
intersect or there is an edge between them.
\end{enumerate}
We say that an embedding of $F$ into $G$ is of {\em depth} $r$ if every vertex of $G$ appears in the image of at most $r$ vertices of $F$. (Thus $F$ has an embedding of depth 1 into $G$ if and only $F$ is a minor of $G$. More generally, if we denote by $G^{(r)}$ the graph obtained by replacing every vertex of $G$ by a clique of size $r$ and replacing every edge by a complete bipartite graph on $r+r$ vertices, then $F$ has an embedding of depth at most $r$ into $G$ if and only if $F$ is a minor of $G^{(t)}$.)

\begin{proposition}[\cite{DBLP:journals/toc/Marx10}]\label{prop:embedding-reduce}
  Let $I$ be a CSP instance with primal graph $F$ and domain $D$.
  Given an arbitrary graph $G$ and an embedding $\phi$ of $F$ into $G$
  having depth at most $r$, we can construct in polynomial time an
  equivalent CSP instance $I'$ having domain $D^r$ and primal graph
  $G$.
\end{proposition}
\ifabstract
One can reduce 3SAT to a CSP instance and, with some additional work, ensure that the primal graph has maximum degree 3.
\else
One can easily reduce 3SAT to a CSP instance and, with some additional work, ensure that the primal graph has maximum degree 3.
\fi
\begin{proposition}[\cite{DBLP:journals/toc/Marx10}]\label{prop:3sat-to-csp}
  Given an instance of 3SAT with $n$ variables and $m$ clauses, it is
  possible to construct in polynomial time an equivalent CSP instance
  with $O(n+m)$ variables, $O(m)$ binary constraints, at most 3
  constraint on each variable, and domain size 3.
\end{proposition}
The Sparsification Lemma of Impagliazzo, Paturi, and Zane \cite{MR1894519} greatly extends the applicability of the Exponential Time Hypothesis: It states that even an $2^{o(m)}$ time algorithm for $m$-clause 3SAT would contradict ETH. Together with Proposition~\ref{prop:3sat-to-csp}, this implies the following lower bound:
\begin{proposition}\label{prop:csp-eth}
  If there is a $2^{o(m)}$ time algorithm for binary CSP instances
  with $m$ constraints, at most 3 constraints on each variable, and
  domain size 3, then ETH fails.
\end{proposition}

What is actually proved by Alon and Marx \cite{DBLP:journals/siamdm/AlonM11} is the following embedding result:
\begin{theorem}[Alon and Marx~\cite{DBLP:journals/siamdm/AlonM11}]\label{th:embedhamming}
  For integers $n, d > 0$ and every graph $F$ with $m > m_0(n, d)$
  edges and no isolated vertices, there is an embedding of depth
  $O(dm/n^d)$ from $F$ into $\Hs[n,d]$.
\end{theorem}
Then Proposition~\ref{prop:csp-eth},
Proposition~\ref{prop:embedding-reduce}, and
Theorem~\ref{th:embedhamming} together show that an $|I|^{o(n^d)}$
time algorithm for CSP on $\Hs[n,d]$ would violate ETH. Additionally, if
$m$ is sufficiently large, then any $f(|V|)$ is dominated by $|I|$ (see
\cite{DBLP:journals/toc/Marx10}).

With Theorem~\ref{th:hamminglowerbound} and
Proposition~\ref{prop:embedding-reduce} at hand, we can prove
Theorem~\ref{th:gridnolog} by finding a suitable embedding from
$d$-dimensional Hamming graphs to $(d+1)$-dimensional grid graphs.
\begin{theorem}\appstar
\ifabstract
\footnote{Due to space restrictions, proofs marked with \appstar appear in the full version of this paper.}
\else
\footnote{Proofs marked with \appstar have been moved to the appendix due to space restrictions.}
\fi
\label{th:hammingtogrid}
For every $d,n\ge 1$ there is an embedding $\phi$ from $\Hs[n,d]$ to $\Rs[n,d+1]$ having depth at most $d+1$. Furthermore, such an embedding can be constructed in time polynomial in $n$ and $d$.
\end{theorem}
\ifabstract
Proposition~\ref{prop:embedding-reduce},
Theorem~\ref{th:hammingtogrid},
and
Theorem~\ref{th:hamminglowerbound} together prove
Theorem~\ref{th:gridnolog}.
\fi

\fi
\iffull
\begin{proof}
\fi
\ifabstract
%Due to lack of space, the proof of Theorem~\ref{th:hamminglowerbound} is given in Section \ref{sec:hamminglowerbound}.

\fi
\ifmove
\ifappendix\section{Proof Theorem~\ref{th:hamminglowerbound}}
\label{sec:hamminglowerbound}
\begin{proof}[Proof of Theorem~\ref{th:hamminglowerbound}]
\fi

The embedding $\phi$ is constructed as follows. For every $\ba=(a_1,\dots, a_d)\in [n]^d$, we define the mapping $\phi_0$, mappings $\phi_i$ for $1\le i \le d$, and mapping $\phi$ as follows:
\begin{align*}
\phi_0(\ba)&=\{ (a_1,\dots,a_d,x)\mid x\in [n] \}\\
\phi_i(\ba)&=\{ (a_1,\dots,a_{i-1},x,a_{i+1},\dots, a_d,a_i)\mid x\in [n] \}\\
\phi(\ba)&=\phi_0(\ba)\cup \bigcup_{i=1}^d \phi_i(\ba)
\end{align*}
Let us verify first that $\phi(\ba)$ is connected for every
$\ba=(a_1,\dots,a_d)\in [n]^d$. It is clear that $\phi_i(\ba)$ is
connected for every $0 \le i \le d$. Moreover, $(a_1,\dots,a_d,a_i)\in \phi_0(\ba)\cap \phi_i(\ba)$, hence the union $\phi(\ba)$ of the $d+1$ sets is also connected.

Let $\ba=(a_1,\dots, a_d)\in [n]^d$ and $\ba'=(a_1,\dots,
a_{i-1},a'_i,a_{i+1},\dots,a_d)\in [n]^d$ be two adjacent vertices of
$\Hs[n,d]$. We have to show that $\phi(\ba)$ and $\phi(\ba')$
touch. Indeed, we have that $(a_,\dots,a_{i-1},a_i,a_{i+1},\dots,
a_d,a'_i)$ appears both in $\phi_0(\ba)$ and in $\phi_i(\ba')$,
implying that $\phi(\ba)$ and $\phi(\ba')$ intersect.

Finally, we have to show that the depth of $\phi$ is at most $d+1$. Consider a vertex $(a_1,\dots, a_{d+1})\in [n]^{d+1}$ of $\Rs[n,d+1]$. This vertex appears in $\phi_0(\ba)$ for exactly one $\ba$ (namely, for $\ba=(a_1,\dots,a_d)$) and, for every $1\le i \le d$, it appears in $\phi_i(\ba)$ for exactly one $\ba$ (namely, for $\ba=(a_1,\dots, a_{i-1},a_{d+1}, a_{i+1},\dots, a_d\}$). Therefore, each vertex of $\Rs[n,d+1]$ appears in $\phi(\ba)$ for exactly $d+1$ values of $\ba\in [n]^d$.
\end{proof}
\fi

\ifmain
\iffull
Proposition~\ref{prop:embedding-reduce},
Theorem~\ref{th:hammingtogrid},
and
Theorem~\ref{th:hamminglowerbound} together prove
Theorem~\ref{th:gridnolog}.
\fi
\fi

\iffull
\subsection{Proof of Theorem~\ref{th:gridbounded-domain}}
\label{sec:proof-domain}
\fi
\ifmove
\ifappendix
	\label{sec:proof-domain}
\fi

%(outline)

%\section{Preliminaries}

%\paragraph{Multi-commodity flows and sparsest-cuts.}
Let $G$ ba an $n$-vertex graph.
A \emph{multi-commodity flow instance} on $G$ is a pair $\mu=(\capa,\dem)$, where $\capa:E(G)\to \mathbb{R}_{\geq 0}$, and $\dem:V(G) \times V(G) \to \mathbb{R}_{\geq 0}$.
For every edge $e\in E(G)$, we say that $\capa(e)$ is its \emph{capacity}, and for every $u,v\in V(G)$ we say that there is $\dem(u,v)$ \emph{demand}  between $u$ and $v$.
We may also similarly define capacities over vertices, in which case the capacity function $\capa$ will also be defined on $V(G)$, i.e.~we have $\capa:V(G)\cup E(G)\to \mathbb{R}_{\geq 0}$, and if there are no edge capacities, we may have $\capa:V(G)\to \mathbb{R}_{\geq 0}$.

%\[
%\phi(S) := \frac{|E(S, V \setminus S)|}{\min\{|S|, |V \setminus S|\}}.
%\]
%We also define 
%\[
%\phi(G) := \min_{S\subset V(G): S \neq \emptyset} { \phi(S) }.
%\]

A \emph{routing} of a multi-commodity flow instance $(\capa, \dem)$ is a a collection of flows ${\cal F} = \{f_{u,v}\}_{u,v\in V(G)}$, where for every $u,v\in V(G)$, $f_{u,v}$ is a flow from $u$ to $v$ of value $\dem(u,v)$.
The \emph{congestion} of an edge $e\in E(G)$ in the routing ${\cal F}$ is the ratio of the total flow passing through $e$, over the capacity of $e$.
The (edge) congestion of ${\cal F}$ is the maximum congestion over all edges.
When we have capacities on the vertices, the congestion of a vertex $u$ is analogously defined to be the ratio of the total flow passing through $u$, over the capacity of $u$; the \emph{vertex congestion} of ${\cal F}$ is similarly the maximum congestion over all vertices.

%We say that ${\cal F}$ \emph{satisfies all capacity constrains} if it has congestion at most 1.

The randomized-rounding technique  of Raghavan and
Thompson \cite{DBLP:journals/combinatorica/RaghavanT87} allows us to
modify a routing of a multi-commodity flow in such a way that every demand is
routed along a single path and vertex congestion\footnote{The result of Raghavan \&
  Thompson \cite{DBLP:journals/combinatorica/RaghavanT87} is stated
  for directed graphs with edge capacities, but it is well-known that
  vertex capacities can be reduced to that case.} increases only at most
logarithmically.  Moreover, it has been shown by Raghavan that this
routing can be constructed in deterministic polynomial time via the
technique of pessimistic estimators (see
\cite{DBLP:journals/jcss/Raghavan88}, \cite{DBLP:journals/combinatorica/RaghavanT87}).

\begin{theorem}[Raghavan \cite{DBLP:journals/jcss/Raghavan88}]\label{thm:RT}
  For any $n$-vertex graph $G$ and for any multi-commodity flow
  instance $\mu$ on $G$, if $\mu$ can be routed with vertex congestion
  at most $c$, then there exists routing of $\mu$ such that every
  demand is routed along a single path and vertex congestion is at
  most $O\left(\frac{c\cdot \log n}{\log\log n}\right)$. Moreover, such a routing
  can be found in deterministic polynomial-time.
\end{theorem}

%%% Local Variables: 
%%% mode: latex
%%% TeX-master: "main"
%%% End: 

The following embedding result is the main technical part of this
section.  More specifically, it gives an embedding of arbitrary graphs
into a $d$-dimensional Hamming graph.  The proof uses mostly
elementary arguments, together with the routing result of Raghavan and
Thompson (Theorem \ref{thm:RT}).  We also give an alternative proof of
a this embedding, albeit with a slightly worse depth, in Appendix~\ref{sec:embedding_alternative}; we remark that the worse bound on the
depth is still sufficient for our application here.  This different
proof combines bounds on the algebraic connectivity of $d$-dimensional
Hamming graphs, together with the duality between multi-commodity
flows and sparsest-cuts. The advantage of that proof is that, unlike
the proof below, it does not require any construction specific to the
graph $\Hs[n,d]$: the result seems to follow from high-level routing
arguments (altough one needs to review some background to exploit
these arguments).

\begin{theorem}\label{thm:embedding}
Let $d\geq 1$.
For every graph $G$ with $m$ edges, no isolated vertices, and with maximum degree $\dmax$, there is an embedding $\phi$ from $G$ to $\Hs[n,d]$ having depth $O\left(\frac{d \cdot \dmax \cdot  \log m}{\log\log m}\right)$, where $n=\lceil m^{1/d} \rceil$.
Moreover, such an embedding can be found in deterministic polynomial time.
\end{theorem}

\begin{proof}
Since $G$ has no isolated vertices, we have $|V(G)| \leq 2m$.
Recall that we have identified $V(\Hs[n,d])$ with $[n]^d$, in the obvious way.
Let $f:V(G) \to [n]^d$ be an arbitrary map that sends at most $2$ vertices in $V(G)$ to any vertex in $[n]^d$.
Define a multi-commodity flow instance 
$\mu = (\capa, \dem)$
on $\Hs[n,d]$, where every edge and vertex of $\Hs[n,d]$ has unit capacity.
and for every $\uu,\vv\in [n]^d$, we have
\[
\dem(\uu,\vv) = |\{\{u',v'\}\in E(G) : f(u')=\uu \text{, and } f(v')=\vv \}|.
\]

That is, $\dem(\uu,\vv)$ is the number of edges that go between vertices mapped to $\uu$ and vertices mapped to $\vv$.
We first define an auxiliary collection of paths in $G$.
For every $\xx,\yy\in [n]^d$, we define a path $P_{\xx,\yy}$ between $\xx$ and $\yy$ in $G$ as follows.
Let $\xx,\yy\in [n]^d$, with $\xx=(x_1,\ldots,x_d)$, $\yy=(y_1,\ldots,y_d)$.
The path $P$ visits the vertices $\zz^1,\ldots,\zz^d$, in this order, where for any $i\in [d]$, we have
\[
\zz^i = (x_1,\ldots,x_i,y_{i+1},\ldots,y_d).
\]
It is immediate that for any $i\in [d-1]$, either $\zz^i=\zz^{i+1}$, or $\{\zz^i,\zz^{i+1}\} \in E(\Hs[n,d])$.
Therefore, the sequence of vertices $\zz^1,\ldots,\zz^d$ defines a valid path $P_{x,y}$ from $\xx$ to $\yy$ in $\Hs[n,d]$.

We proceed to define a routing of the instance $(\capa, \dem)$.
Consider an edge $\{u,v\}\in E(G)$.
For every $\xx\in [n]^d$, let $Q_{\xx}$ be the path obtained by concatenating $P_{f(u),\xx}$ and $P_{\xx,f(v)}$.
Let $F_{u,v,\xx}$ be the flow obtained by routing $1/n^d$ demand units along $Q_{\xx}$.
Let $F_{u,v} = \sum_{\xx\in [n]^d} F_{u,v,\xx}$, i.e.~the flow obtained by summing up all these flows for all $\xx\in [n]^d$.
Clearly, the flow $F_{u,v}$ sends 1 demand unit from $f(u)$ to $f(v)$.
Taking the union of all such flows $F_{u,v}$ for all edges $\{u,v\}\in E(G)$, we obtain a multi-commodity flow $F$ that is a routing of $(\capa, \dem)$.

We now bound the congestion of $F$.
Let $\uu\in [n]^d$.
We first bound the number of ordered pairs $(\mya,\myb)\in [n]^d \times [n]^d$, such that $\uu\in V(P_{\mya,\myb})$.
Let $i\in \{1,\ldots,d\}$, and let us first bound the number of ordered pairs $(\mya,\myb)$ with $\mya=(a_1,\dots,a_d)$ and $\myb=(b_1,\dots,b_d)$ such that $\uu$ is the $i$-th vertex visited by $P_{\mya,\myb}$; we may assume that each $P_{\mya,\myb}$ visits exactly $d$ vertices, possibly with repetitions.
This happens precisely when 
\[
\uu = (a_1,\ldots,a_i,b_{i+1},\ldots,b_d).
\]
Therefore, there are $n^{d-i}$ possible choices for $\mya$ (by
choosing the values of $a_{i+1}$, $\dots$, $a_d$), and $n^i$ possible
choices for $\myb$ (by choosing the values of $b_1$, $\dots$,
$b_{i}$).  In total, there are at most $n^d$ possible choices for
$(\mya,\myb)$ such that $\uu$ is the $i$-th vertex visited by
$P_{\mya,\myb}$.  It follows that there are at most $d\cdot n^d$
possible pairs $(\mya,\myb)$ such that $\uu$ is visited by
$P_{\mya,\myb}$. 

In the routing $F$ constructed above, there are at most $2$ vertices
of $G$ mapped to every vertex of $\Hs[n,d]$.  Every such vertex has at
most $\dmax$ incident edges in $G$, and therefore $F$ routes at most
$2 \dmax$ unit flows between every vertex $\ww$ and the rest of $G$.
Each such unit flow sends at most $1/n^d$ units of flow along each
path $P_{\mya,\myb}$ with $\mya=\ww$, and at most $1/n^d$ units along
each path $P_{\mya,\myb}$ with $\myb=\ww$.  It follows that when
constructing the routing $F$, for each $(\mya,\myb)\in [n]^d \times
[n]^d$, we send $1/n^d$ units of flow along $P_{\mya,\myb}$ at most
$2\cdot (2 \dmax) = 4 \dmax$ times.  Since every $\uu\in [n]^d$ is
visited by at most $d\cdot n^d$ paths, we conclude that the total flow
in $F$ that passes through $\uu$ is at most $4 d \dmax$.  As every
vertex has unit capacity, it follows that the maximum vertex
congestion in $F$ is at most $4 d \dmax$.

By Theorem~\ref{thm:RT}, we can compute in deterministic polynomial time
a routing such that every demand is routed along a single
path, and with maximum vertex congestion at most $O\left(\frac{d \cdot
    \dmax \cdot \log m}{\log\log m}\right)$.  We can now construct
an embedding $g$ of $G$ into $\Hs[n,d]$ as follows.  For very $v\in
V(G)$, we set $g(v)$ to be the union of all the vertices in all the
paths used to route the edges in $G$ that are incident to $v$.  It is
immediate that $g$ is a valid embedding with depth $O\left(\frac{d
    \cdot \dmax \cdot \log m}{\log\log m}\right)$, as required.
\end{proof}

The depth of the embedding given by Theorem~\ref{thm:embedding} has a
factor logarithmic in the size of $G$. Therefore, an application of
Proposition~\ref{prop:embedding-reduce} with this embedding cannot
give a CSP whose domain size is constant. Therefore, in the following
lemma, we decrease the depth by roughly a factor of $\beta$ by
``stretching'' the graph $\Hs[n,d]$ by a factor $\beta$ in all
dimensions.

\begin{lemma}\label{lem:hamming_depth}
Let $n,d\geq 1$.
Let $G$ be a graph, and let $f$ be an embedding of $G$ into $\Hs[n,d]$, with depth $\alpha$.
Let $\beta \leq \alpha$.
There exists an embedding $f'$ of $G$ into $\Hs[n\cdot \beta, d]$, with depth $d\cdot \lceil \alpha/\beta\rceil$.
Moreover, $f'$ can be computed in deterministic polynomial time.
\end{lemma}
\begin{proof}
For any $u\in V(G)$, and for any $\vv\in f(u)$, pick a value $\tau(u,\vv)\in \{0,\ldots,\beta-1\}$, such that for any $\vv\in [n]^d$, and for any $t\in \{0,\ldots,\beta-1\}$, there exist at most $\lceil \alpha/\beta\rceil$ vertices $u\in V(G)$, with $\tau(u,\vv)=t$.
Such a choice of values clearly exists, since for every $\vv\in [n]^d$, there exist at most $\alpha$ vertices $u\in V(G)$, with $\vv\in f(u)$.

For any $t\in \{0,\ldots,\beta-1\}$, and $i\in [d]$, we define 
\[
X(t,i) = \{(x_1,\ldots,x_d)\in \{0,\ldots,\beta-1\}^d : x_i = t\}.
\]
We also set
\[
X(t) = \bigcup_{i\in [d]} X(t,i).
\]
That is, $X(t)$ is the set of vectors in $\{0,\ldots,\beta-1\}^d$
having value $t$ in at least one of the coordinates.  We define the
desired embedding $f'$ as follows.  For any $u\in V(G)$, let
\[
f'(u) = \bigcup_{\vv\in f(u)} \left(\beta\cdot \vv - X(\tau(u,\vv))\right),
\]
%where $\mathbf{1}_d$ denotes the all-ones vector in $\mathbb{Z}$,
where for any $\vv\in \mathbb{Z}^d$, and $X\subseteq \mathbb{Z}^d$, we have used the notation $\vv+X = \bigcup_{\ww\in X} (\vv+\ww)$, and $\vv+\ww$ denotes vector addition.
This completes the definition of the embedding $f'$.

Let us first verify that the depth of $f'$ is at most $d\cdot \lceil \alpha/\beta\rceil$.
Let $\ww = (w_1,\ldots,w_d)\in [\beta\cdot n]^d$.
There exists a unique $\vv=(v_i,\ldots,v_d)\in [n]^d$, such that for any $i\in [d]$, we have $\beta (v_i-1)+1 \leq w_i \leq \beta v_i$.
Note that for any $\uu\in V(G)$, we have $\ww\in f'(u)$ only if $\vv\in f(u)$.
There exist at most $\alpha$ vertices in $\uu\in V(G)$, such that $\vv\in f(u)$.
For every $i\in [d]$, there exists exactly one value $t\in \{0,\ldots,\beta-1\}$, such that 
\[
\ww \in \beta\cdot \vv - X(t,i).
\]
Indeed, every vector in $X(t,i)$ has the same value $t$ at the $i$-th
coordinate, hence $t$ has to be the $i$-th coordinate of $\beta\cdot
\vv-\ww$.  For every such value $t$, there exist at most $\lceil
\alpha/\beta\rceil$ vertices $u\in V(G)$ with $\vv\in f(u)$ and
$\tau(u,\vv)=t$.  Therefore, there exist at most $d\cdot \lceil
\alpha/\beta\rceil$ vertices $u\in V(G)$ with $\ww\in f'(u)$, as
required.

It remains to argue that $f'$ is a valid embedding, that is, for any $\{u,u'\}\in V(G)$, the sets $f'(u)$ and $f'(u')$ touch.
Suppose first that $f(u)\cap f(u') \neq \emptyset$.
Since for any $t,t'\in \{0,\ldots,\beta-1\}$ we have $X(t)\cap X(t') \neq \emptyset$, it follows that $f'(u)\cap f'(u') \neq \emptyset$.
The only remaining case is that there exists $\vv\in f(u)$ and $\vv'\in f(u')$, such that $\{\vv,\vv'\}\in E(\Hs[n,d])$.
Let $\vv=(v_1,\ldots,v_d)$, and $\vv'=(v_1',\ldots,v_d')$.
It follows that there exists $i\in [d]$, such that $\vv'=(v_1,\ldots,v_{i-1},v_i',v_{i+1},\ldots,v_d)$.
This implies that for any $t,t'\in \{0,\ldots,\beta-1\}$, there exists an edge between the sets
$\beta \cdot \vv - X(t,i)$
and
$\beta \cdot \vv' - X(t',i)$: there are two vectors in $X(t,i)$ and $X(t',i)$ that differ only in the $i$-th coordinate, hence there 
are two vectors in $\beta \cdot \vv - X(t,i)$
and
$\beta \cdot \vv' - X(t',i)$ that differ only in the $i$-th coordinate.
This establishes that $f'$ is a valid embedding, concluding the proof.
\end{proof}

Combining Theorem \ref{thm:embedding} \& Lemma \ref{lem:hamming_depth}, we immediately obtain the following.

\begin{theorem}\label{thm:embedding_2}
Let $d\geq 1$.
For every graph $G$ with $m$ edges, no isolated vertices, and with maximum degree $\dmax$, there is an embedding $\phi$ from $G$ to $\Hs[n,d]$ having depth $O\left(d^2 \cdot \dmax\right)$, where $n=O\left( m^{1/(d-1)} \cdot \frac{\log m}{\log\log m} \right)$.
Moreover, such an embedding can be found in deterministic polynomial time.
\end{theorem}
We are now ready to prove Theorem~\ref{th:gridbounded-domain}.
\begin{proof}[Proof of Theorem \ref{th:gridbounded-domain}]
Consider a CSP instance $I$ with primal graph $G$, with $m$ edges, maximum degree $3$, and domain size $3$.
%Assuming ETH, by Proposition \ref{prop:csp-eth} we have that there exists no algorithm for all such $I$ with running time $2^{o(m)}$.
By Theorem \ref{thm:embedding_2}, we can compute in polynomial time an embedding $f$ of $G$ into $\Hs[n,d-1]$ with $n=\left(m^{1/d-1} \cdot \frac{\log m}{\log\log m}\right)$
and depth $O(d^2)$.
By Theorem \ref{th:hammingtogrid} we can compute in polynomial time an embedding $f'$ of $\Hs[n,d-1]$ into $\Rs[n,d]$ with depth $O(d)$.
By composing $f$ and $f'$, we obtain a polynomial-time computable embedding $g$ of $G$ into $\Rs[n,d]$, with depth $O(d^3)$.
By Proposition \ref{prop:embedding-reduce} we can compute in polynomial time a CSP $I'$ that is equivalent with $I$, having primal graph $\Rs[n,d]$, and with domain size $3^{O(d^3)}$.
Note that for any $\eps>0$, we have 
\ifabstract
\begin{align*}
|V(\Rs[n,d])|^{1-1/d-\eps} &= \left(m^{1/(d-1)}\cdot \frac{\log m}{\log\log m}\right)^{d\cdot (1-1/d-\eps)}\\
 &=
m^{1-d\eps/(d-1)}\cdot \left( \frac{\log m}{\log\log m}\right)^{1-1/d-\eps}= o(m).
\end{align*}
\else
\[
|V(\Rs[n,d])|^{1-1/d-\eps} = \left(m^{1/(d-1)}\cdot \frac{\log m}{\log\log m}\right)^{d\cdot (1-1/d-\eps)}=
m^{1-d\eps/(d-1)}\cdot \left( \frac{\log m}{\log\log m}\right)^{1-1/d-\eps}= o(m).\]
\fi
It follows that if for some $\eps>0$ there exists a $2^{O(|V|^{1-1/d-\eps})}$ algorithm for CSP$({\cal R}_d)$ with domain size at most $s_d=3^{O(d^3)}$, then there exists a $2^{o(m)}$ algorithm for $I$, which by Proposition \ref{prop:csp-eth} contradicts ETH.
\end{proof}

%%% Local Variables: 
%%% mode: latex
%%% TeX-master: "main"
%%% End: 

\fi

\fi

\ifmain
\subsection{Geometric CSP problems}
\label{sec:geo-csp}
In this section, we give lower bounds for CSP problems of certain
special forms that are particularly suited for reductions to
$d$-dimensional geometric problems. First, we consider CSPs where
every constraint satisfies the following property: we say that a
constraint $\langle (u,v),R\rangle$ is a {\em projection from $u$ to
  $v$} if for every $x\in D$, there is at most one $y\in D$ such that
$(x,y)\in R$ (a projection from $v$ to $u$ is defined analogously).
In other words, the relation $R$ is of the form $\{(x,y)\mid y=p(x)\}$
for some function $p$; we call this function $p$ the {\em projection
  function} associated with the projection constraint.  A CSP instance
is a {\em projection CSP} if every binary constraint is a projection;
in addition to the binary constraint, the instance may contain any
number of unary constraints.  Note that an edge $uv$ in the
undirected primal graph of a projection CSP does not tell us whether the corresponding
constraint is a projection from $u$ to $v$, or a projection from $v$
to $u$.

\begin{proposition}\appstar\label{prop:project}
  There is a polynomial time algorithm that, given a CSP instance $I$
  on $\Rs[n,d]$ and domain $D$, creates an equivalent projection CSP
  instance $I'$ on $\Rs[2n,d]$ and domain $D^2$.
\end{proposition}
\fi

\ifmove

\ifappendix\section{Proofs from Section~\ref{sec:constr-satisf-probl}}\begin{proof}[Proof (of Proposition~\ref{prop:project})]\else
\begin{proof}\fi
 We construct the instance
  $I'$ the following way. For every $\ba\in [n]^d$, if there is a
  unary constraint on $\ba$ in $I$, then we introduce into $I'$ the same unary
  constraint on $2\ba$ (that is, on the vector obtained from $\ba$ by multiplying each coordinate by 2). Consider now a binary constraint $\langle
  \ba,\ba',R\rangle$ where $\ba=(a_1,\dots,a_d)$, $\ba'=\ba+\be_i$, and $R\subseteq
  D^2$. Let
  $\bb=\ba+\ba'=(2a_1,\dots,2a_{i-1},2a_i+1,2a_{i+1},\dots,2a_d)$; observe that
  $\bb$ is adjacent to both $2\ba$ and $2\ba'$ in $\Rs[2n,d]$. We
  introduce the unary constraint $\langle (\bb),R \rangle$ to $I'$
  (note that the domain of $I'$ is $D^2$, thus $R$ can be interpreted as a valid unary
  constraint relation in $I'$). Moreover, we introduce the binary constraints
\begin{gather*}
\left\langle ( 2\ba,\bb) , \{ ((x,1), (x,y) \mid x,y\in D \}\right\rangle\\
\left\langle ( 2\ba',\bb) , \{ ((y,1), (x,y) \mid x,y\in D \}\right\rangle.
\end{gather*}
That is, the first constraint ensures that the first coordinate of the
value of $2\ba$ is the same as the first coordinate of the value of
$\bb$ and the second coordinate of $2\ba$ is forced to the dummy value
$1$. Similarly, the first coordinate of $2\ba'$ is the same as the
second coordinate of $\bb$. This ensures that if $\ba$ receives the
value $(x,1)$ and $\ba'$ receives the value $(y,1)$, then $\bb$
receives $(x,y)$; taking into account the unary constraint on $\bb$,
it follows that $(x,y)\in R$. Effectively, we have implemented the
binary constraint relation $R$ on (the first coordinates of) $2\ba$ and $2\ba'$ using two
projections and a unary relation. 

Let the constraints of $I'$ introduced above be called the
``important'' constraints of $I'$ and let the variables in these
constraints be the important variables of $I'$.  For every edge $uv$ of
$\Rs[2n,d]$ for which we do not yet have a corresponding constraint in
$I'$, we proceed as follows. At most one of $u$ and $v$ can be an
important variable (as we have defined above an important constraint between any pair of adjacent important variables); suppose that $v$ is not an important
variable. Then we introduce the constraint
\[
\langle (u,v),  \{ ((x,y),(1,1))\mid (x,y)\in D^2 \} \rangle.
\]
Clearly, this constraint is a projection from $u$ to $v$. Therefore, the resulting instance $I'$ is a projection CSP with primal graph $\Rs[2n,d]$.

We claim that $I$ has a satisfying assignment if and only if $I'$
has. Let $f'$ be a satisfying assignment of $I'$. By our discussion
above, if we define $f(\ba)=f'(2\ba)$ for every $\ba\in [n]^d$, then
$f$ is a satisfying assignment of $I$. Conversely, suppose that $f$ is
a satisfying assignment of $I$. We define a satisfying assignment of
$I'$ as follows. For every $\ba\in [n]^d$, we set
$f'(2\ba)=f(\ba)$. For every constraint $\langle (\ba,\ba'),R\rangle$
of $I$ with $\ba'=\ba+\be_i$, we set
$f(\ba+\ba')=(f(\ba),f(\ba'))$. This way, we assigned a value to every
important variable; let us assign $(1,1)$ to the variables that are
not important. It is easy to verify that this gives a satisfying
assignment of $I'$.
\end{proof}
\fi

\ifmain
A {\em $d$-dimensional geometric $\le$-CSP} is a CSP of the following
form.  The set $V$ of variables is a subset of the vertices of
$\Rs[n,d]$ for some $n$ and the primal graph is an induced subgraph of
$\Rs[n,d]$ (that is, if two variables are adjacent in $\Rs[n,d]$, then
there is a binary constraint on them).  The domain is $[\dommax]^d$
for some integer $\dommax\ge 1$. The instance can contain arbitrary
unary constraints, but the binary constraints are of a special form.
A {\em geometric constraint} is a constraint $\langle
(\ba,\ba'),R\rangle$ is with $\ba'=\ba+\be_i$ such that
\ifabstract
$R=\{ ((x_1,\dots,x_d),(y_1,\dots,y_d))\mid \text{$x_i\le y_i$} \}$.
\else
\[
R=\{ ((x_1,\dots,x_d),(y_1,\dots,y_d))\mid \text{$x_i\le y_i$} \}.
\]
\fi
In other words, if $\ba$ and $\ba'$ are adjacent with $\ba'$ being
larger by one in the $i$-th coordinate, then the $i$-th coordinate of
the value of $\ba'$ should be at least as large as the $i$-th
coordinate of the value of $\ba$.
 Note that a $d$-dimensional geometric CSP
is fully defined by specifying the set of variables and the unary
constraints: the binary constraints of the
instance are then determined by the definition.

\begin{proposition}\appstar\label{prop:csp-to-geo}
  For every $d\ge 2$, given a projection CSP instance $I$ on $\Rs[n,d]$
  and domain $D$, one can construct in polynomial time an equivalent
  $d$-dimensional geometric $\le$-CSP instance $I'$ with domain $[2|D|+1]^d$ and $O(n^d)$
  variables.
\end{proposition}
\fi
\ifmove
\ifappendix\begin{proof}[Proof (of Proposition~\ref{prop:csp-to-geo})]\else
\begin{proof}\fi
  For convenience of notation, we assume that the domain of $I$ is
  $D=[\dommax]$ and we construct a CSP $I'$ with domain
  $\{-\dommax,\dots,\dommax\}^d$. Then by adding $\dommax+1$ to every
  coordinate, we can get an equivalent instance with domain
  $[2\dommax+1]^d$. 
 
For every variable $\ba=(a_1,\dots,a_d)\in [n]^d$  of $I$, we define
\begin{equation}
\ba^*=(5a_1,5a_2,2a_3,\dots,2a_d)\label{eq:stardef}
\end{equation}
and
introduce into $I'$ the 12 variables listed in Figure~\ref{fig:12var}, forming a cycle of length 12 in the first two dimensions.

Let $G$ be the primal graph of $I$. As every constraint of $I$ is a
projection constraint, we can obtain an orientation $\dirG$ of $G$
such that if $\dirG$ has a directed edge from $u$ to $v$, then $I$ has
a projection constraint from $u$ to $v$ (if the constraint on $u$ and
$v$ happens to be both a projection from $u$ to $v$ and a projection
from $v$ to $u$, then we may choose the orientation of the edge
between $u$ and $v$ arbitrarily).  Suppose that $\dirG$ has a directed edge
from $\ba$ to $\ba+\be_1$, which means that there is a projection
constraint from $\ba$ to $\ba+\be_1$; let function $p:D\to D$ the be
the projection function associated with the constraint.  Then we define
$q^+_1(\ba,x)=p(x)$ for every $x\in D$. If there is no such directed
edge from $\ba$ to $\ba+\be_1$, then we define $q^+_1(\ba,x)=x$.  We
define $q^-_1(\ba,x)$ similarly, depending on whether there is a
directed edge from $\ba$ to $\ba-\be_1$: if there is directed edge
from $\ba$ to $\ba-\be_1$ and $p$ is the projection function associated with
the projection constraint, then $q^-_1(\ba,x)=p(x)$; otherwise, we let
$q^-_1(\ba,x)=x$.  Note that an assignment $f$ satisfies the
constraint on $\ba$ and $\ba+\be_1$ if and only if
\begin{equation}\label{eq:projcsp}
q^+_1(\ba,f(\ba))=q^-_1(\ba+\be_1,f(\ba+\be_1))
\end{equation}
holds.  For example, in the case when there is an edge from $\ba$ to
$\ba+\be_1$ in $\dirG$ and $p$ is the projection function of the
corresponding projection constraint, then we defined
$q^+_1(\ba,f(\ba))=p(f(\ba))$ and
$q^-_1(\ba+\be_1,f(\ba+\be_1))=f(\ba+\be_1)$, and the projection
constraint is satisfied if and only if $p(f(\ba))=f(\ba+\be_1)$.
Similarly, in the case when there is an edge from $\ba+\ba_1$ to $\ba$
in $\dirG$ and $p$ is the projection function of the corresponding
projection constraint, then we defined $q^+_1(\ba,f(\ba))=f(\ba)$ and
$q^-_1(\ba+\be_1,f(\ba+\be_1))=p(f(\ba+\be_1))$, and the projection
constraint is satisfied if and only if $f(\ba)=p(f(\ba+\be_1))$.
Intuitively, the gadget representing variable $\ba$ should propagate
the value $q^+_1(\ba,f(\ba))$ towards the gadget representing
$\ba+\be_1$, which in turn should propagate the value
$q^-_1(\ba+\be_1,f(\ba+\be_1))$ towards the gadget of $\ba$. Then the
validity of the constraint between $\ba$ and $\ba+\be_1$ can be tested
by testing whether these two values are equal.

  The
functions $q^+_2$, $q^-_2$ are defined similarly based on the
existence of directed edges between $\ba$ and $\ba+\be_2$, and between
$\ba$ and $\ba-\be_2$ in $\dirG$.
For coordinates $3\le i \le d$, we define the analogous functions in a slightly different way. Let us define the following functions:
\begin{gather*}
m^\top_i(\ba)=\begin{cases}
\ba+\be_i & \text{if the $i$-th coordinate of $\ba$ is odd,}\\
\ba-\be_i & \text{if the $i$-th coordinate of $\ba$ is even.}
\end{cases}\\
m^\bot_i(\ba)=\begin{cases}
\ba-\be_i & \text{if the $i$-th coordinate of $\ba$ is odd,}\\
\ba+\be_i & \text{if the $i$-th coordinate of $\ba$ is even.}
\end{cases}
\end{gather*}
Note that $m^\top_i(m^\top_i(\ba))=m^\bot_i(m^\bot_i(\ba))=\ba$ for
every vector $\ba$.  Suppose that $\dirG$ has a directed edge from
$\ba$ to $m^\top_i(\ba)$; let $p$ be the projection function
associated with the projection constraint from $\ba$ to
$m^\top_i(\ba)$. Then we define $q^\top_i(\ba,x)=p(x)$. If there is
no such directed edge, then we define $q^\top_i(\ba,x)=x$. Again, we get that 
the constraint on $\ba$ and $m^\top_i(\ba)$ is satisfied if only if
\begin{equation}\label{eq:projcsp2}
q^\top_i(\ba,f(\ba))=q^\top_i(m^\top_i(\ba),f(m^\top_i(\ba)))
\end{equation}
holds.  We define $q^\bot_i(\ba,x)$ similarly, by considering the
directed edge from $\ba$ to $m^\bot_i(\ba)$.

We may assume that every variable $\ba$ of $I$ has a unary constraint
$\langle (\ba) ,R_\ba\rangle$ on it. We introduce a corresponding
unary constraint on the 12 variables created above the following way.
For every variable $\ba$ and $x\in R_\ba$, we add the tuple shown in Figure~\ref{fig:12var} to the unary constraint of each of the 12 variables introduced above.

We represent the binary constraints of $I$ the following way.  For
every pair $\ba$ and $\ba+\be_i$ of adjacent variables of $I$, we
introduce two new variables into $I'$.  Let $\ba^*$ defined as above.
\begin{itemize}
\item If $i=1$, then we introduce the variables $\ba^*+\be_1-\be_2$ and $\ba^*+\be_1-2\be_2$.
\item If $i=2$, then we introduce the variables $\ba^*-2\be_1+\be_2$ and $\ba^*-\be_1+\be_2$.
\item If $i>2$ and $\ba+\be_i=m^\top_i(\ba)$, then we introduce the variables $\ba^*-3\be_1+\be_i$ and $\ba^*+\be_i$.
\item If $i>2$ and $\ba+\be_i=m^\bot_i(\ba)$, then we introduce the variables $\ba^*-3\be_1-3\be_2+\be_i$ and $\ba^*-3\be_2+\be_i$.
\end{itemize}
These newly introduced variables do not have any unary constraints on
them, that is, they can receive any value from the domain.

\begin{figure}
\newcommand{\cell}[5]{
\fbox{\parbox{3.3cm}{
\begin{center}$#1$
\end{center}

\noindent
\begin{multline*}
(#2, #3,
 #4,\\
\dots, #5)
\end{multline*}
}}}
\noindent
{\small
\begin{tabular}{cccc}
\cell{\ba^*-3\be_1}{x}{x}{q^\top_3(\ba,x)}{q^\top_d(x)}&
\cell{\ba^*-2\be_1}{x}{q^+_2(\ba,x)}{0}{0}&
\cell{\ba^*-\be_1}{x}{-q^+_2(\ba,x)}{0}{0}&
\cell{\ba^*}{x}{-x}{-q^\top_3(\ba,x)}{-q^\top_d(x)}\\[17mm]
\cell{\ba^*-3\be_1-\be_2}{q^-_1(\ba,x)}{x}{0}{0}&&&
\cell{\ba^*-\be_2}{q^+_1(\ba,x)}{-x}{0}{0}\\[17mm]
\cell{\ba^*-3\be_1-2\be_2}{-q^-_1(\ba,x)}{x}{0}{0}&&&
\cell{\ba^*-2\be_2}{-q^+_1(\ba,x)}{-x}{0}{0}\\[17mm]
\cell{\ba^*-3\be_1-3\be_2}{-x}{x}{q^\bot_3(\ba,x)}{q^\bot_d(x)}&
\cell{\ba^*-2\be_1-3\be_2}{-x}{q^-_2(\ba,x)}{0}{0}&
\cell{\ba^*-\be_1-3\be_2}{-x}{-q^-_2(\ba,x)}{0}{0}&
\cell{\ba^*-3\be_2}{-x}{-x}{-q^\bot_3(\ba,x)}{-q^\bot_d(x)}
\end{tabular}
}
\caption{The 12 variables introduced into $I'$ and the form of the
  tuples in the unary constraints (Proof of
  Proposition~\ref{prop:csp-to-geo}).}\label{fig:12var}
\end{figure}

Let us discuss the intuition behind the construction. Each variable of
$I$ is represented by a cycle of 12 variables of $I'$, extending in
the first two dimensions (see Figure~\ref{fig:cycles}). The unary constraints of these 12 variables
restrict the tuples appearing on the variables to have the form shown
in Figure~\ref{fig:12var}. As we shall see,  by observing the first two
coordinates of these tuples and geometric constraints between adjacent
variables in the cycle, we can conclude that the values of these
variables should be coordinated, that is, the same value $x$ should
appear in each of them.
\begin{figure}
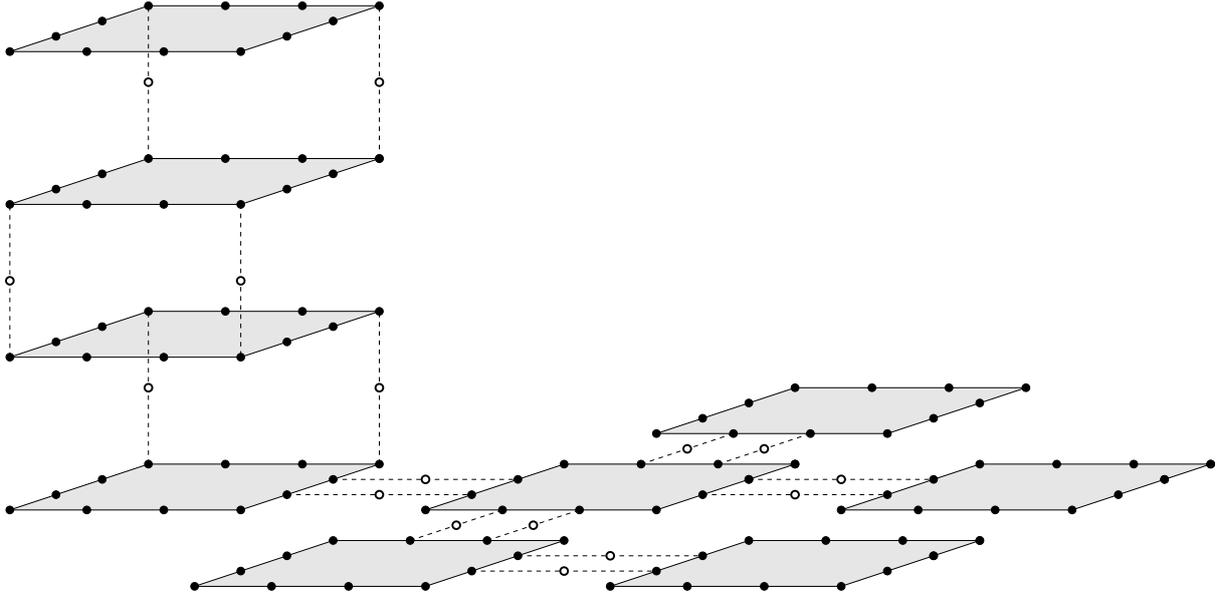

\begin{center}
{\small \svg{\linewidth}{cycles}}
\caption{A an example of the reduction in Proposition~\ref{prop:csp-to-geo} in $d=3$ dimensions. A projection CSP instance $I$ with 9 variables and 8 binary constraints is reduced to a geometric CSP instance with 9 cycles of 12 variables each (black dots represent the variables, normal lines represent the binary constraint between them) and 16 connection variables (white dots) with 32 binary constraint between the connection variables and the cycles (dashed lines).}\label{fig:cycles}
\end{center}
\end{figure}

Each binary projection constraint is
represented by connecting the two cycles with two new variables. Very
briefly, these connections work the following way. Suppose that there
is a projection constraint between $\ba$ and $\ba+\be_i$ forcing that
that $p(f(\ba))=f(\ba+\be_i)$ and let $\bb^+$ and $\bb^-$ be the two
connection variables. Then the cycle corresponding to $\ba$ has a variable
adjacent to $\bb^+$ whose value is ``shifted towards'' $\bb^+$ by
exactly $p(f(\ba))$ and there is a variable corresponding to
$\ba+\be_i$ that is adjacent to $\bb^+$ and whose value is ``shifted
away'' from $\bb^+$ by exactly $f(\ba+\be_i)$. Therefore, we can
conclude $p(f(\ba))\le f(\ba+\be_i)$. Similarly, by observing the
values of the two variables adjacent to $\bb^-$, we can conclude
$p(f(\ba))\ge f(\ba+\be_i)$ and $p(f(\ba))= f(\ba+\be_i)$ follows.

The choice of the location for the connecting variables $\bb^+$ and
$\bb^-$ merits a brief explanation. For $i=1$, the connection between
cycles representing $\ba$ and $\ba+\be_i$ fairly obvious: we connect
the right side of the cycle of $\ba$ with the left side of the cycle
of $\ba+\be_i$ (see Figure~\ref{fig:cycles}). Similarly, for $i=2$, we
connect the upper side of $\ba$ with the lower side of
$\ba+\be_i$. However, for $i\ge 3$, the connections are slightly
different. The issue is that we do not want a variable of the cycle of
$\ba$ to take part both in the connection to $\ba+\be_i$ and in the
connection to $\ba-\be_i$. Therefore, if the $i$-th coordinate of
$\ba$ is odd (that is, if $m^{\top}_i(\ba)=\ba+\be_i$), then we
connect the upper side of the cycle of $\ba$ with the upper side of
the cycle of $\ba+\be_i$; if the $i$-th coordinate of $\ba$ is even
(that is, if $m^{\bot}_i(\ba)=\ba+\be_i$), then we connect the lower
side of the cycle of $\ba$ with the lower side of the cycle of
$\ba+\be_i$.

Let $f$ be a satisfying assignment of $I$. We construct a satisfying
assignment $f'$ of $I'$ the following way. Consider a variable $\ba$
of $I$ and let $\ba^*$ defined as in \eqref{eq:stardef}. Let $x=f(\ba)$. For the 12
variables of Figure~\ref{fig:12var} corresponding to $\ba$, we set
the value of $f'$ to the vectors shown in the figure. By checking the
first two coordinates of these vectors, we can verify that every binary
geometric constraint between these 12 variables is satisfied (in
fact, with equality).

Let us extend now $f'$ to the variables introduced to represent the
binary constraints of $I$. Suppose that $\ba$ and $\ba+\be_i$ are
adjacent variables. If $i=1$, then we have introduced variables
$\ba^*+\be_1-\be_2$ and $\ba^*+\be_1-2\be_2$; we let $x=f(\ba)$ and
define $f'(\ba^*+\be_1-\be_2)=(q^+_1(\ba,x),0,\dots,0)$ and
$f'(\ba^*+\be_1-2\be_2)=(-q^+_1(\ba,x),0,\dots,0)$.  Then the first
coordinates of $f'(\ba^*-\be_2)$ and $f'(\ba^*+\be_1-\be_2)$ are the
same.  Moreover, as $f$ satisfies the constraint on $\ba$ and
$\ba+\be_1$, we have by \eqref{eq:projcsp} that
$q^+_1(\ba,f(\ba))=q^-_1(\ba+\be_1,f(\ba+\be_1))$. Therefore, the
first coordinate of $f'(\ba^*+\be_1-\be_2)$ is the same as the first
coordinate of
$f'((\ba+\be_1)^*-3\be_1-\be_2)=f'(\ba^*+2\be_1-\be_2)$. Similarly,
the first coordinate of $f'(\ba^*+\be_1-2\be_2)$ is the same as the
first coordinates of $f'(\ba^*-2\be_2)$ and
$f'((\ba+\be_1)^*-3\be_1-2\be_2)=f'(\ba^*+2\be_1-2\be_2)$. Finally, it
is clear that $f'(\ba^*+\be_1-\be_2)$ and $f'(\ba^*+\be_1-2\be_2)$ are
compatible with each other, as the second coordinate is 0 in both of
them.

For $i=2$, we extend $f'$ by defining
$f'(\ba^*-2\be_1+\be_2)=(0,q^+_2(\ba,x),0,\dots,0)$ and
$f'(\ba^*-\be_1+\be_2)=(0,-q^+_2(\ba,x),0,\dots,0)$; an argument similar to the case $i=1$
shows that these values are compatible with the assignments of $f'$
defined so far. For $i\ge 3$ and $\ba+\be_i=m^\top_i(\ba)$, we set
\begin{alignat*}{1}
f'(\ba^*-3\be_1+\be_i)&=(\overbrace{0,\dots, 0}^{i-1},q^\top_i(\ba,x),\overbrace{0,\dots,0}^{d-i}),\\
f'(\ba^*+\be_i)&=(\overbrace{0,\dots,0}^{i-1}, -q^\top_i(\ba,x),\overbrace{0,\dots,0}^{d-i})
\end{alignat*}
(that is, only the $i$-th coordinates of these
vectors are nonzero).  As $f$ satisfies the constraint on $\ba$ and $\ba+\be_i$,  we have by \eqref{eq:projcsp2} that
$q^\top_i(\ba,f(\ba))=q^\top_i(\ba+\be_1,f(\ba+\be_1))$. Then we can argue
that $f'$ on $\ba^*-3\be_1+\be_i$ and $\ba^*+\be_i$ are compatible with the
assignments on their neighbors. The argument is similar for the case
$i\ge 3$ and $\ba+\be_i=m^\bot_i(\ba)$. Therefore, we have shown that a
satisfying assignment $f'$ for $I'$ exists.

Assume now that $f'$ is a satisfying assignment for $I'$. Let $\ba$ be a variable of $I$ and consider the 12 corresponding variables in Figure~\ref{fig:12var}. 
The value of $f'$ is of the form shown in Figure~\ref{fig:12var} (as these are the only vectors allowed by the unary constraints on these variables), but the value of $x$ can be different for each of these 12 variables. Let $x_{\ba^*}$, $x_{\ba^*-\be_1}$, etc.~be the values of $x$ appearing on the variables. As $f'$ satisfies the geometric constraints between adjacent variables, we have the following inequalities:
\begin{alignat*}{4}
x_{\ba^*-3\be_1}&\le 
x_{\ba^*-2\be_1}&&\le 
x_{\ba^*-\be_1}&&\le 
x_{\ba^*} && \text{  (upper side)}\\
-x_{\ba^*-3\be_2}& \le
-x_{\ba^*-2\be_2}&& \le
-x_{\ba^*-\be_2}&& \le
-x_{\ba^*} &&\text{  (right side)}\\
-x_{\ba^*-3\be_1-3\be_2}&\le 
-x_{\ba^*-2\be_1-3\be_2}&&\le 
-x_{\ba^*-\be_1-3\be_2}&&\le
-x_{\ba^*-3\be_2}&&\text{  (lower side)}\\
x_{\ba^*-3\be_1-3\be_2}&\le
x_{\ba^*-3\be_1-2\be_2}&&\le
x_{\ba^*-3\be_1-\be_2}&&\le
x_{\ba^*-3\be_1}&&\text{  (left side)}
\end{alignat*}

Putting together, we get a cycle of inequalities and hence all the
inequalities above have to be equalities. This means that there is a
common value $x$ such that $f'$ is defined with this value on all 12
of the variables as in Figure~\ref{fig:12var}; let us define $f(\ba)$
to be this value $x$. Clearly, $f$ satisfies the unary constraints on
the variables, as otherwise, say,
$f'(\ba^*)=(x,-x,-q^\top_3(\ba,x),\dots,-q^\top_d(\ba,x))$ would not
satisfy the unary constraints of $I'$ on $\ba^*$. Consider now a binary constraint
on $\ba$ and $\ba+\be_1$ in $I$. The first coordinate of
$f'(\ba^*-\be_2)$ is $q^+_1(\ba,f(\ba))$ and the first coordinate of
$f'((\ba+\be_1)^*-3\be_1-\be_2)=f'(\ba^*+2\be_1-\be_2)$ is
$q^-_1((\ba+\be_1)^*,f(\ba+\be_i))$. If the first coordinate of
$f'(\ba^*+\be_1-\be_2)$ is $x'$, then the geometric constraints ensure
that $q^+_1(\ba^*,f(\ba))\le x' \le
q^-_1((\ba+\be_1)^*,f(\ba+\be_i))$. Similarly, if $x''$ is the first coordinate of
$f'(\ba^*+\be_1-2\be_2)$, then comparing it with the first coordinates of 
$f'(\ba^*-2\be_2)$ and
$f'((\ba+\be_1)^*-3\be_1-2\be_2)$, we get that $-q^+_1(\ba^*,f(\ba)) \le x'' 
\le -q^-_1((\ba+\be_1)^*,f(\ba+\be_i))$. Thus we get
$q^+_1(\ba^*,f(\ba))= q^-_1((\ba+\be_1)^*,f(\ba+\be_i))$, which means,
by \eqref{eq:projcsp}, that $f$ satisfies the binary constraint on
$\ba$ and $\ba+\be_1$. In a similar way, for every $1\le i \le d$ and
binary constraint on $\ba$ and $\ba+\be_i$, we can show that $f$
satisfies this constraint by looking at the two variables of $I'$ that
connect the 12 variables corresponding to $\ba$ with the 12 variables
corresponding to $\ba+\be_i$.
\end{proof}
\fi

\ifmain
Theorem~\ref{th:gridnolog} and Propositions~\ref{prop:project} and \ref{prop:csp-to-geo} imply the following lower bound on geometric CSPs:
\begin{theorem}\label{th:geomcspunbounded}
If for some fixed $d\ge 1$, there is an\\ $f(|V|)n^{o(|V|^{1-1/d})}$ time algorithm for $d$-dimensional geometric $\le$-CSP for some function $f$, then ETH fails.
\end{theorem}

% Note: we could define $=$-CSP, where we have $x_i=y_i$ in the
% definition. Basically without extra work, we could state
% Theorem~\ref{th:geomcspunbounded} for $=$-CSP if this is interesting
% or useful for some of the reductions.

\fi

%%% Local Variables: 
%%% mode: latex
%%% TeX-master: "main"
%%% End: 

\ifmain
\section{Packing problems}
\label{sec:packing}

In this section, we prove the lower bounds on packing $d$-dimensional
unit balls and cubes. These results can be obtained
quite easily from the lower bounds for geometric CSP problems proved
in Section~\ref{sec:geo-csp}. For simplicity of notation, we state the
results for open balls and cubes, but of course the same bounds hold
for closed balls and cubes.
\begin{theorem}
If for some fixed $d\ge 2$, there is an\\ $f(k)n^{o(k^{1-1/d})}$ time algorithm for finding $k$ pairwise nonintersecting $d$\-/dimensional open unit  balls in a collection of $n$ balls, then ETH fails.
\end{theorem}
\begin{proof}
  It will be convenient to work with open balls of diameter
  1 (that is, radius $1/2$) in this proof: then two balls are nonintersecting if
  and only if the distance of their centers are at least 1.  Let $I$
  be a $d$-dimensional $\le$-CSP instance on variables $V$ and domain
  $[\dommax]^d$. We construct a set $B$ of $d$-dimensional balls such
  that $B$ contains a set of $|V|$ pairwise nonintersecting balls if
  and only if $I$ has a satisfying assignment. Therefore, if we can
  find $k$ nonintersecting balls in time $f(k)n^{o(k^{1-1/d})}$, then
  we can solve $I$ in time $f(k)n^{o(|V|^{1-1/d})}$. By
  Theorem~\ref{th:geomcspunbounded}, this would contradict ETH.

  Let $\epsilon=1/(d\dommax^2)$. Let $\ba=(a_1,\dots,a_d)$ be a
  variable of $I$ and let $\langle (\ba),R_\ba\rangle$ be the unary
  constraint on $\ba$. For every $\bx=(x_1,\dots,x_d)\in
  R_\ba\subseteq [\dommax]^d$, we introduce into $B$ an open ball of diameter $\frac{1}{2}$ centered
  at $(a_1+\epsilon x_1,\dots, a_d+\epsilon x_d)=\ba+\epsilon \bx$;
  let $B_\ba$ be the set of these $|R_{\ba}|$ balls. Note that the
  balls in $B_\ba$ all intersect each other. Therefore, $B'\subseteq
  B$ is a set of pairwise nonintersecting balls, then $|B'|\le |V|$
  and $|B'|=|V|$ is possible only if $B'$ contains exactly one ball from each
  $B_\ba$. In the following, we prove that there is such a set of
  $|V|$ pairwise nonintersecting balls if and only if $I$ has a
  satisfying assignment.

  We need the following observation first.  Consider two balls
  centered at $\ba+\epsilon \bx$ and $\ba+\be_i+\epsilon\bx'$ for some
  $\bx=(x_1,\dots,x_d)\in [\dommax]^d$ and $\bx'=(x'_1,\dots,x'_d)\in
  [\dommax]^d$.  We claim that they are nonintersecting if and only if
  $x_i\le x'_{i}$. Indeed, if $x_i>x'_i$, then the square of the
  distance of the two centers is
\ifabstract
\begin{alignat*}{1}
\sum_{j=1}^{i-1} \epsilon^2(x'_j-x_j)^2+&(1+\epsilon(x'_i-x_i))^2+
\sum_{j=i+1}^{d} \epsilon^2(x'_j-x_j)^2\\
 &\le d\epsilon^2\dommax^2+(1+\epsilon(x'_i-x_i))^2\\
 &\le \epsilon  + (1-\epsilon)^2 =1-\epsilon+\epsilon^2 < 1
\end{alignat*}
\else
\begin{alignat*}{1}
\sum_{j=1}^{i-1} \epsilon^2(x'_j-x_j)^2+
(1+\epsilon(x'_i-x_i))^2+
\sum_{j=i+1}^{d} \epsilon^2(x'_j-x_j)^2&\le d\epsilon^2\dommax^2+(1+\epsilon(x'_i-x_i))^2\\&\le
\epsilon  + (1-\epsilon)^2 =1-\epsilon+\epsilon^2 < 1
\end{alignat*}
\fi
(we have used $x'_i,x_i\le \dommax$ in the first inequality and
$\epsilon = 1/(d\dommax^2)$ in the second inequality).  On the other
hand, if $x_i\le x'_i$, then the square of the distance is at least
$(1+\epsilon(x'_i-x_i))^2\ge 1$, hence the two balls do not intersect
(recall that the balls are open).  This proves our claim. Moreover, it
is easy to see that if $\ba$ and $\ba'$ are not adjacent in $\Rs[n,d]$,
then the balls centered at $\ba+\epsilon \bx$ and $\ba'+\epsilon \bx'$
cannot intersect for any $\bx,\bx'\in [\dommax]^d$: the square of the
distance of the two centers is at least $2(1-\epsilon\dommax)^2>1$.

Let $f$ be a satisfying assignment for $I$. For every variable $\ba$,
we select the ball $\ba+\epsilon f(\ba)\in B_\ba$. If $\ba$ and $\ba'$ are not adjacent, then $\ba+\epsilon f(\ba)$ and $\ba'+\epsilon f(\ba')$
cannot intersect. If $\ba$ and $\ba'$ are adjacent, then there is a geometric binary constraint on $\ba$ and
$\ba'$. Therefore, if, say, $\ba'=\ba+\be_i$, then the binary constraint ensures that the
$i$-th coordinate of $f(\ba)$ is at most the $i$-th coordinate of
$f(\ba')$. By our claim in the previous paragraph, it follows that the
balls centered at $\ba+\epsilon f(\ba)$ and $\ba'+\epsilon f(\ba')$ do
not intersect.

Conversely, let $B'\subseteq B$ be a set of $|V|$ pairwise independent
balls. This is only possible if for every $\ba\in V$, set $B'$
contains a ball from $B_\ba$, that is, centered at $\ba+\epsilon
f(\ba)$ for some $f(\ba)\in [\dommax]^d$. We claim that $f$ is a satisfying
assignment of $I$. First, it satisfies the unary constraints: the fact
that $\ba+\epsilon f(\ba)$ is in $B_\ba$ implies that $f(\ba)$ satisfies
the unary constraint on $\ba$. Moreover, let $\ba$ and
$\ba'=\ba+\be_i$ be two adjacent variables.  Then, as we have observed
above, the fact that $\ba+\epsilon f(\ba)$ and $\ba'+\epsilon f(\ba')$
do not intersect implies that the $i$-th coordinate of $f(\ba)$ is at
most the $i$-th coordinate of $f(\ba')$. That is, the geometric binary
constraint on $\ba$ and $\ba'$ is satisfied.
\end{proof}
 
The lower bound for packing $d$-dimensional axis-parallel unit-side cubes
is similar, but there is a slight difference. In the case of unit
balls, as we have seen, balls centered at $\ba+\epsilon \bx$ and
$\ba+\be_i+\be_j+\epsilon\bx'$ cannot intersect (if $\epsilon$ is
sufficiently small), but this is possible for unit cubes. Therefore,
we represent each variable with $2d+1$ cubes: a cube and its two
neighbors in each of the $d$ dimensions.
\begin{theorem}\label{th:packcube}\appstar
If for some fixed $d\ge 2$, there is an $f(k)n^{o(k^{1-1/d})}$ time algorithm for finding $k$ pairwise nonintersecting $d$-dimensional open axis-parallel unit  cubes in a collection of $n$ cubes, then ETH fails.
\end{theorem}
\fi

%\ifmove
%\ifappendix\section{Proof of Theorem~\ref{th:packcube}}\fi
\begin{proof}
  It will be convenient to work with open cubes of side length 1 in
  this proof: then two cubes are nonintersecting if and only if the
  centers differ by at least 1 in at least one of the coordinates.
  Let $I$ be a $d$-dimensional $\le$-CSP instance on variables $V$ and
  domain $[\dommax]^d$. We construct a set $B$ of $d$-dimensional
  axis-parallel cubes such that $B$ contains a set of $(2d+1)|V|$
  pairwise nonintersecting cubes if and only if $I$ has a satisfying
  assignment. Therefore, if we can find $k$ nonintersecting cubes in
  time $f(k)n^{o(k^{1-1/d})}$, then we can solve $I$ in time
  $f(k)n^{o(|V|^{1-1/d})}$. By Theorem~\ref{th:geomcspunbounded}, this
  would contradict ETH.

  Let $\epsilon=1/(2\dommax)$. Let
  $\ba=(a_1,\dots,a_d)$ be a variable of $I$ and let $\langle
  (\ba),R_\ba\rangle$ be the unary constraint on $\ba$. For every
  $\bx=(x_1,\dots,x_d)\in R_\ba\subseteq [\dommax]^d$, we introduce into
  $B$ the cube centered at $3\ba+\epsilon \bx$, and for every $1\le i
  \le d$, the two cubes centered at $3\ba+\be_i+\epsilon \bx$ and
  $3\ba-\be_i+\epsilon \bx$.  Let us call $B_{\ba,\bx}$ this set of
  $2d+1$ cubes.  Note that the $2d+1$ cubes in $B_{\ba,\bx}$ do not
  intersect each other (recall that the cubes are open). Moreover, at
  most $2d+1$ pairwise nonintersecting cubes can be selected from
  $B_{\ba}:=\bigcup_{\bx\in R_\ba} B_{\ba,\bx}$: for example, the
  cubes centered at $\ba+\be_i+\epsilon \bx$ and $\ba+\be_i+\epsilon
  \bx'$ intersect for any $\bx,\bx'\in[\dommax]^d$.

  Let $f$ be a satisfying assignment for $I$. We show that selecting
  the $2d+1$ cubes $B_{\ba,f(\ba)}$ for each variable $\ba\in V$ gives
  a set of $(2d+1)|V|$ pairwise nonintersecting cubes.  We claim that
  if $\ba$ and $\ba'$ are not adjacent, then the cubes in
  $B_{\ba,f(\ba)}$ do not intersect the cubes in $B_{\ba',f(\ba')}$.
  First, if $\ba$ and $\ba'$ differ by at least 2 in some coordinate,
  then $3\ba$ and $3\ba'$ differ by at least 6 in that coordinate and
  then it is clear that, e.g., the cubes centered at
  $3\ba+\be_i+\epsilon f(\ba)$ and $3\ba'-\be_j+\epsilon f(\ba')$
  cannot intersect (note that $\epsilon f(\ba)\le \epsilon \dommax \le
  1/2$). On the other hand, if $\ba$ and $\ba'$ differ in at least two
  coordinates, then $3\ba$ and $3\ba$ differ by at least 3 in those
  two coordinates and $3\ba+\epsilon f(\ba)$ and $3\ba'+\epsilon
  f(\ba')$ differ by at least $3-\epsilon \dommax \ge 2$ in those two
  coordinates. Then adding two unit vectors to the centers cannot
  decrease both differences to strictly less than 1, that is,
  $3\ba+\be_i+\epsilon f(\ba)$ and $3\ba+\be_j+\epsilon f(\ba)$ differ
  by at least 1 in at least one coordinate for any $1\le i,j\le
  d$. This means that the cubes in $B_{\ba,f(\ba)}$ do not intersect
  the cubes in $B_{\ba',f(\ba')}$. Consider now two variables $\ba$
  and $\ba'$ that are adjacent; suppose that $\ba'=\ba+\be_i$. Then
  there is a geometric binary constraint on $\ba$ and $\ba'$, which
  ensures that the $i$-th coordinate of $f(\ba)$ is at most the $i$-th
  coordinate of $f(\ba')$. It follows that $3\ba+\be_i+\epsilon
  f(\ba)\in B_{\ba,f(\ba)}$ does not intersect $3\ba'-\be_i+f(\ba')\in
  B_{\ba',f(\ba')}$. Furthermore, the other
  cubes in $B_{\ba,f(\ba)}$ have $i$-th coordinate less than the
  $i$-th coordinate of $3\ba+\be_i+\epsilon f(\ba)$ and the other cubes in $B_{\ba',f(\ba')}$ have $i$-th coordinate greater than the $i$-th coordinate of $3\ba'-\be_i+f(\ba')$, hence they cannot
  intersect either.

  Conversely, let $B'\subseteq B$ be a set of $(2d+1)|V|$ pairwise
  independent cubes. This is only possible if for every $\ba\in V$,
  set $B'$ contains $2d+1$ cubes selected from $B_\ba$; in particular,
  $B'$ contains a cube centered at $3\ba+\epsilon f(\ba)$ for some
  $f(\ba)\in [\dommax]^d$. We claim that $f$ is a satisfying assignment
  of $I$. First, it satisfies the unary constraints: the fact that
  $3\ba+\epsilon f(\ba)$ is in $B$ implies that $f(\ba)$ satisfies the
  unary constraint on $\ba$. Moreover, let $\ba$ and $\ba'=\ba+\be_i$ be two
  adjacent variables. Then, as $B'$ contains $2d+1$ cubes from each of $B_\ba$ and $B_{\ba'}$,
  it has to contain cubes $3\ba+\be_i+\epsilon\bx_1$ and
  $3\ba'-\be_i+\epsilon\bx_2=3\ba+2\be_i+\epsilon\bx_2$ for some
  $\bx_1,\bx_2\in [\dommax]^d$. Now the $i$-th coordinate of
  $3\ba+\be_i+\epsilon\bx_1$ cannot be less than the $i$-th
  coordinate of $3\ba+\epsilon f(\ba)$, that is, the $i$-th coordinate
  of $\bx_1$ is at least the $i$-th coordinate of $f(\ba)$. Similarly,
  the $i$-th coordinate of $\bx_2$ is at least the $i$-coordinate of
  $\bx_1$, and the $i$-th coordinate of $f(\ba')$ is at least the
  $i$-th coordinate of $\bx_2$. Putting together, we get that the
  $i$-coordinate of $f(\ba)$ is at least the $i$-th coordinate of
  $f(\ba)$, that is, the geometric binary constraint on $\ba$ and
  $\ba'$ is satisfied.
\end{proof}
%\fi

%%% Local Variables: 
%%% mode: latex
%%% TeX-master: "main"
%%% End: 

\ifmain
\section{The reduction to TSP}
\label{sec:TSP}
\fi
\ifappendix
\ifabstract
\section{The reduction to TSP: details\\ omitted from Section~\ref{sec:TSP}}
\else
\section{The reduction to TSP: details omitted from Section~\ref{sec:TSP}}
\fi
\fi

\ifmain
Let $d$ be a positive integer.
An instance of \emph{TSP in $d$-dimensional Euclidean space} is a pair $\psi=(X, \alpha)$, where $X$ is a finite set of points in $\mathbb{R}^d$, and $\alpha>0$ is an integer; the goal is to decide whether the length of the shortest TSP tour for $X$ is at most $\alpha$.
%We may assume that all the coordinates of the points in $P$ are polynomially-large (in $|P|$) integers.

Our reduction is inspired by the NP-hardness proof of TSP in $\mathbb{R}^2$ due to Papadimitriou \cite{DBLP:journals/tcs/Papadimitriou77}.
We remark that our proof critically assumes $d\geq 3$.
However, even though we don't know how to make our argument work in $\mathbb{R}^2$, we can still recover the desired lower bound on the running time for $d=2$ by the reduction of \cite{DBLP:journals/tcs/Papadimitriou77}.

We begin by introducing some terminology that will allow us to construct the desired TSP instances.
Some of the definitions are from \cite{DBLP:journals/tcs/Papadimitriou77}. However, some of them have been extended for our setting. In particular, we use a more elaborate construction that takes advantage of the fact that $d\geq 3$.

\textbf{1-Chains.}
Let $x,y\in \mathbb{R}^d$.
A \emph{$1$-chain} from $x$ to $y$ is a sequence of points $\{x_i\}_{i=1}^{k}$, with $x_1=x$, $x_k=y$, such that for any $i\in [k-1]$, the points $x_i$ and $x_{i+1}$ differ in exactly one coordinate, and $\|x_i-x_{i+1}\|_1 = 1$ (see Figure \ref{fig:1_chain}).
We also require that for any $i,j\in [k]$, with $|i-j|\leq 20$, we have $\|x_i-x_j\|_1 \geq \|i-j\|$.

\begin{figure}[t]
\begin{center}
\ifabstract
\scalebox{0.6}{\includegraphics{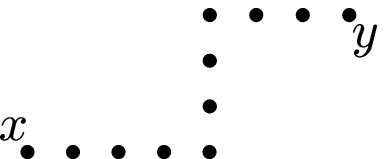}}
\else
\scalebox{0.8}{\includegraphics{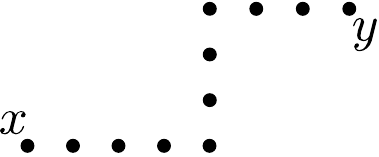}}
\fi
\caption{A 1-chain from $x$ to $y$.\label{fig:1_chain}}
\end{center}
\end{figure}

\textbf{2-Chains.}
Let $x,y\in \mathbb{R}^d$, and let $\theta \in [0, \pi/100)$.
A \emph{$\theta$-ribbon} from $x$ to $y$ is a pair $(P, {\cal H})$, where $P$ is a simple polygonal curve with endpoints $x$ and $y$, consisting of $k$ line segments $s_1,\ldots,s_k$, and ${\cal H}=\{H_i\}_{i=1}^k$ is a family of 2-dimensional planes in $\mathbb{R}^d$,  satisfying the following conditions:
\begin{description}
\item{(I)}
For any $i,j\in [k]$, with $|i-j|>20$, for any $p\in s_i$, $q\in s_j$, we have $\|p-q\|_2 > 40$.
\item{(II)}
The segments $s_1$, and $s_k$ are of length $1$.
All other segments are of length $2$.
\item{(III)}
For every segment $s_i$, we have $s_i\subset H_i$.
\item{(IV)}
For every two consecutive segments $s_i$, $s_{i+1}$, at least  one of the following conditions holds:
  \begin{description}
  \item{(IV-1)}
  Let $p$ be the common endpoint of $s_i$, and $s_{i+1}$, and let $\ell$ bet the line in $H_i$ passing through $p$, and being normal to $s_i$.
  Then the 2-plane $H_{i+1}$ is obtained by rotating $H_{i}$ around $\ell$ by some angle of at most $\theta$.
  \item{(IV-2)}
  The segments $s_i$ and $s_{i+1}$ are collinear, and $H_{i+1}$ is obtained by rotating $H_i$ around $s_i$ by some angle of at most $\theta$.
  \end{description}
\end{description}
Given a $\theta$-ribbon $R=(P,{\cal H})$ from $x$ to $y$, we define a set of points $Y=Y(R)$, which we refer to as a \emph{$2$-chain} with \emph{angular defect $\theta$} (or simply a $2$-chain when $\theta$ is clear from the context) from $x$ to $y$, corresponding to $R$.
We set 
$Y := \{x,y\} \cup \bigcup_{i=1}^{k-1} \{p_i, q_i\}$,
where for any $i$, the points $p_i,q_i \in \mathbb{R}^d$ are defined as follows:
Let $a$ be the common endpoint of $s_{i}$, and $s_{i+1}$.
Let $\ell$ be the line in $H_{i+1}$ normal to $s_{i+1}$ that passes through $a$.
Let $p_i$, $q_i$ be the two points in $\ell$ that are at distance $1/2$ from $a$.
We assign $p_i, q_i$ so that for any two consecutive segments $s_j,s_{j+1}$, with $j\in \{1,\ldots,k-1\}$, the angle between the vectors $p_i-q_i$, and $p_{i+1}-q_{i+1}$, is at most $\theta$ (this is always possible since $(P, {\cal H})$ is a $\theta$-ribbon).
Notice that there are precisely two distinct possibilities of assigning the points $p_1,\ldots,p_{k-1}$, and $q_1,\ldots,q_{k-1}$.
We refer to the points $\{p_i\}_{i=1}^k$ as the \emph{left side}, and the points $\{q_i\}_{i=1}^k$ as the \emph{right side} of the 2-chain.
This concludes the definition of a $2$-chain (see Figure \ref{fig:2_chain}).

\ifabstract
\begin{figure}[t]
\begin{center}
\scalebox{0.5}{\includegraphics{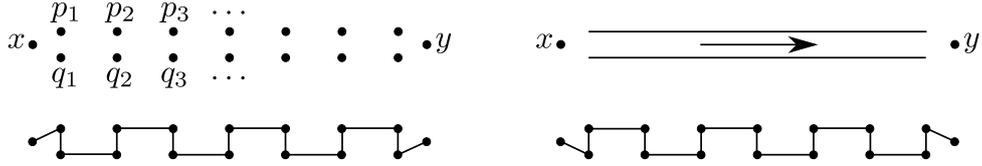}}
\caption{A 2-chain from $x$ to $y$ (top left), its schematic abbreviation (top right), and the two possible optimal paths: mode 1 (bottom left), and mode 2 (bottom right).\label{fig:2_chain}}
\end{center}
\end{figure}
\else
\begin{figure}[t]
\begin{center}
\scalebox{0.8}{\includegraphics{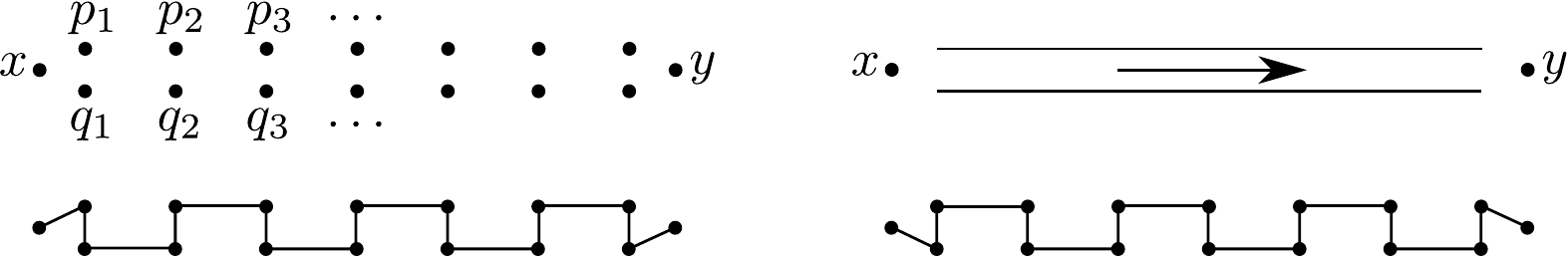}}
\caption{A 2-chain from $x$ to $y$ (top left), its schematic abbreviation (top right), and the two possible optimal paths: mode 1 (bottom left), and mode 2 (bottom right).\label{fig:2_chain}}
\end{center}
\end{figure}
\fi

\begin{lemma}\label{lem:two_chain}
There exists a constant $\theta^*>0$, such that the following holds.
Let $x,y\in \mathbb{R}^d$, 
and let $Y$ be a $2$-chain from $x$ to $y$,
and with angular defect $\theta^*$.
Let $\{p_i\}_{i=1}^{k-1}$, and $\{q_i\}_{i=1}^{k-1}$ be the left, and right sides of $Y$ respectively.
Then, there are precisely two optimal Traveling Salesperson paths from $x$ to $y$ for the set $Y$:
the first one is $x p_1 q_1 q_2 p_2 p_3 q_3 \ldots y$, and the second one is $x q_1 p_1 p_2 q_2 q_3 p_3 \ldots y$.
In the former case we say that $Y$ is traversed in \emph{mode 1}, and in the latter case in \emph{mode 2}.
\end{lemma}

\textbf{The configuration-$\mathbf{H}$.}
We recall the following gadget from \cite{DBLP:journals/tcs/Papadimitriou77}.
A set of points $Y \subset \mathbb{R}^d$ is called a \emph{configuration-$\mathbf{H}$} if there exists a $2$-plane $h$ containing $Y$, so that $Y$ on $h$ appears as in Figure \ref{fig:configuration_H}\footnote{The distance that is set to 8 in Figure \ref{fig:configuration_H}, was set to $6$ in the original construction from \cite{DBLP:journals/tcs/Papadimitriou77}. This appears to be a minor error in \cite{DBLP:journals/tcs/Papadimitriou77}.}.

\ifabstract
\begin{figure}[t]
\begin{center}
\scalebox{0.45}{\includegraphics{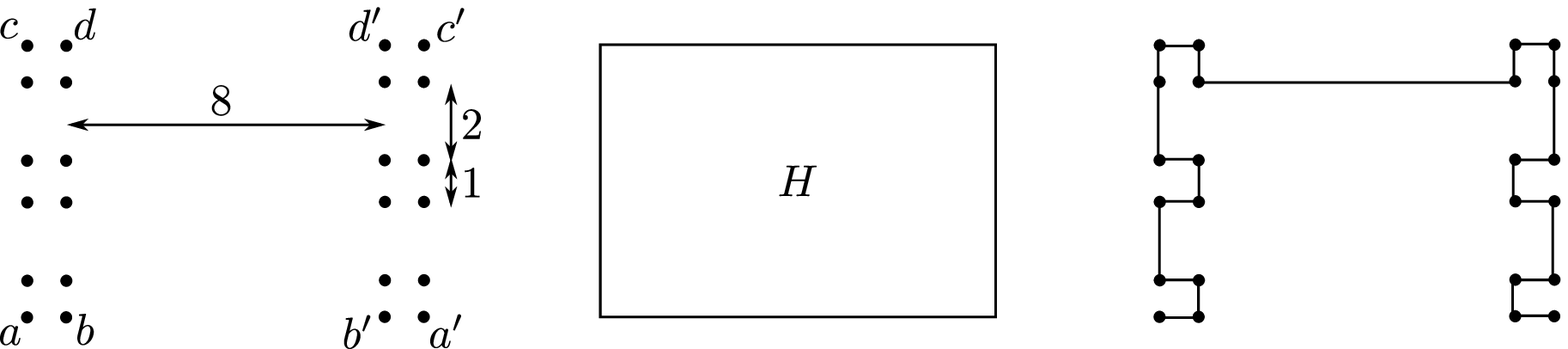}}
\caption{A configuration-$\mathbf{H}$ (left), its schematic abbreviation (middle), and an optimal path from $a$ to $a'$ (right).\label{fig:configuration_H}}
\end{center}
\end{figure}
\else
\begin{figure}[t]
\begin{center}
\scalebox{0.7}{\includegraphics{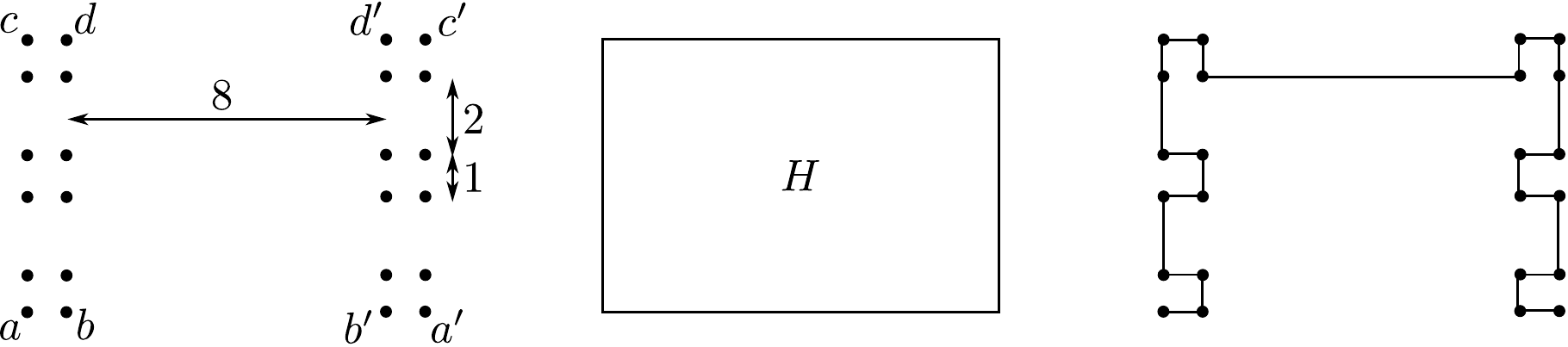}}
\caption{A configuration-$\mathbf{H}$ (left), its schematic abbreviation (middle), and an optimal path from $a$ to $a'$ (right).\label{fig:configuration_H}}
\end{center}
\end{figure}
\fi

\begin{lemma}[Papadimitriou \cite{DBLP:journals/tcs/Papadimitriou77}]\label{lem:configuration_H}
Among all Traveling Salesperson paths having as endpoints two of the points in $a$, $a'$, $b$, $b'$, $c$, $c'$, $c$, $d$ and $d'$, there are $4$ optimal paths with length 32, namely those with endpoints $(a,a')$, $(b,b')$, $(c,c')$, $(d,d')$.
\end{lemma}

Let $(P, {\cal H})$ be a $\theta$-ribbon from $x$ to $y$, with ${\cal H}=\{H_i\}_{i=1}^k$, and let $Y$ be the corresponding $2$-chain.
Suppose that there exists some \emph{odd} $j\in \{1,\ldots,k-15\}$, such that for any $r,r'\in \{j,\ldots,j+14\}$, we have $H_r=H_{r'}$.
Let $Z$ be a configuration-$\mathbf{H}$ contained in the $2$-plane $H_j$.
Suppose further that $Z \cup Y$ appear in $H_j$ as in Figure \ref{fig:neighbor_H} (top).
Then, we say that the configuration-$\mathbf{H}$ $Z$ is a \emph{left neighbor} of the 2-chain $Y$ at $j$.
Similarly, we define a \emph{right neighbor} of a 2-chain (see bottom of Figure \ref{fig:neighbor_H}).
%Moreover, if $Z$ is closer to $p_j$ than $q_j$, then we say that it is a \emph{left} neighbor of $Z$, and otherwise we say that it is a \emph{right} neighbor.

\ifabstract
\begin{figure}
\begin{center}
\scalebox{0.36}{\includegraphics{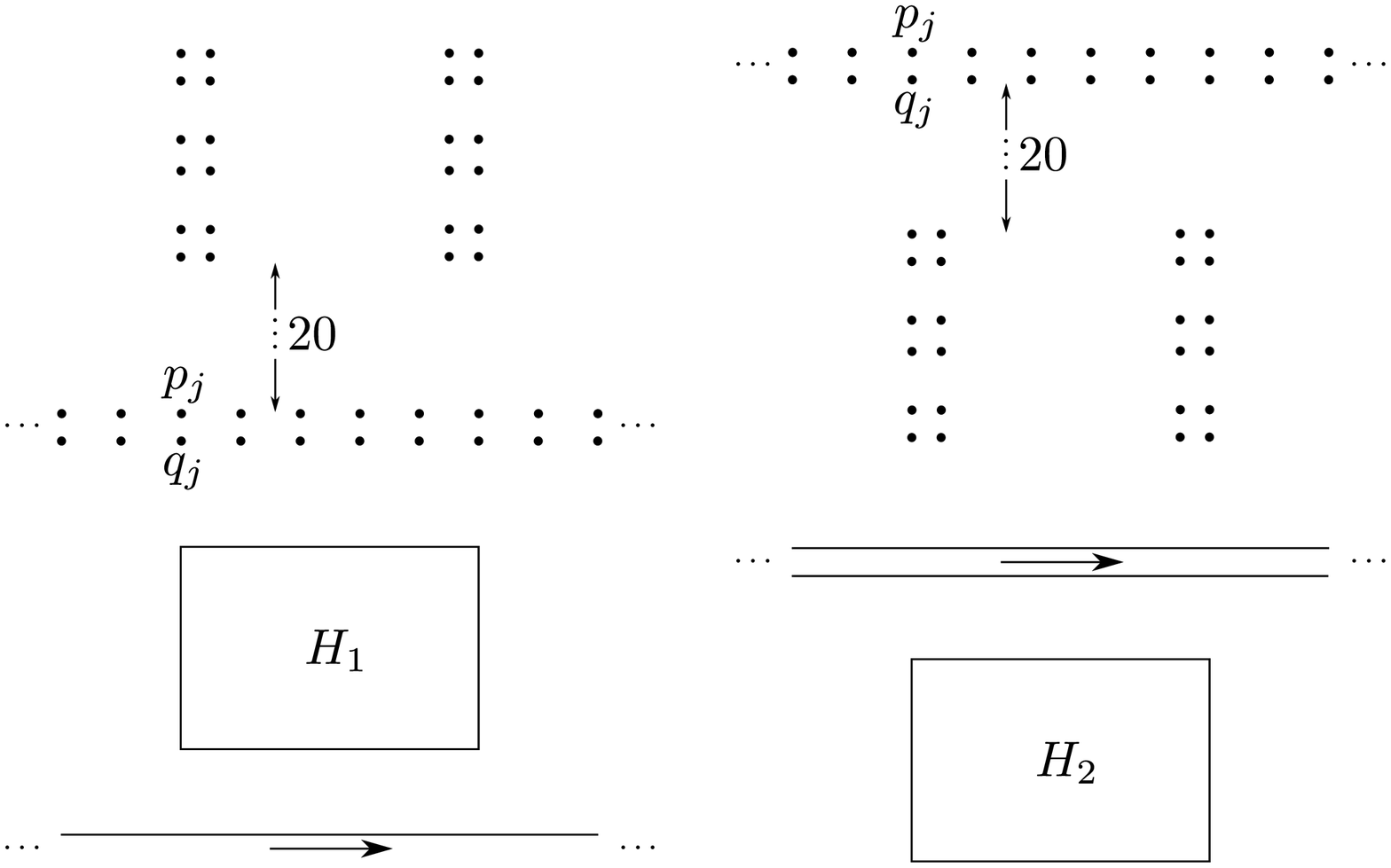}}
\caption{A configuration-$\mathbf{H}$ $H_1$ that is a left neighbor of a 2-chain, and its schematic abbreviation (left).
A configuration-$\mathbf{H}$ $H_2$ that is a right neighbor of a 2-chain, and its schematic abbreviation (right).\label{fig:neighbor_H}}
\end{center}
\end{figure}
\else
\begin{figure}[t]
\begin{center}
\scalebox{0.65}{\includegraphics{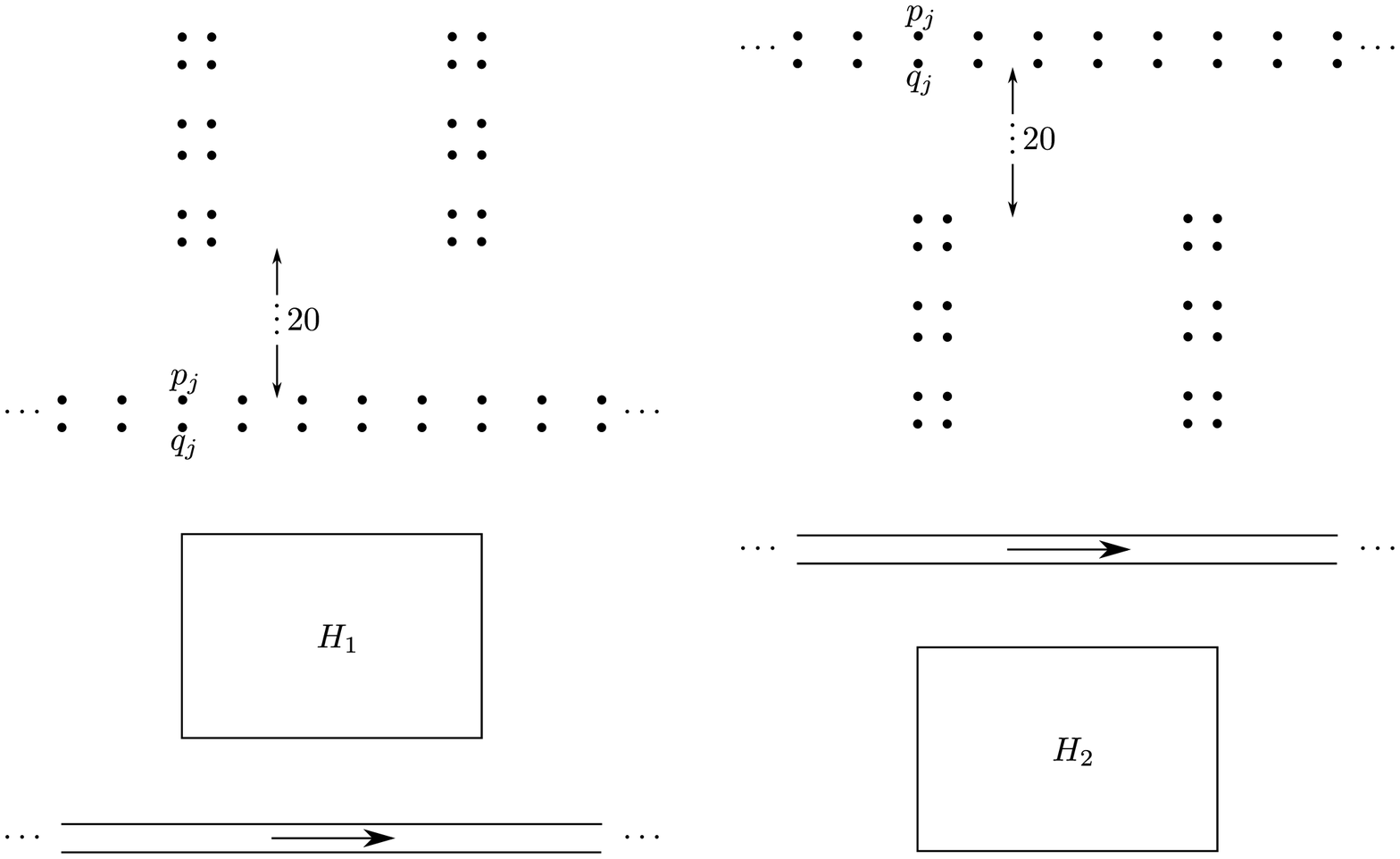}}
\caption{A configuration-$\mathbf{H}$ $H_1$ that is a left neighbor of a 2-chain, and its schematic abbreviation (left).
A configuration-$\mathbf{H}$ $H_2$ that is a right neighbor of a 2-chain, and its schematic abbreviation (right).\label{fig:neighbor_H}}
\end{center}
\end{figure}
\fi

\textbf{The configuration-$\mathbf{B}$.}
Following \cite{DBLP:journals/tcs/Papadimitriou77}, we say that a set of points $Y \subset \mathbb{R}^d$ is a \emph{configuration $B$} if there exists a $2$-plane $h$ containing $Z$, so that $Z$ on $h$ appears as in Figure \ref{fig:configuration_B}.

\ifabstract
\begin{figure}[t]
\begin{center}
\scalebox{0.5}{\includegraphics{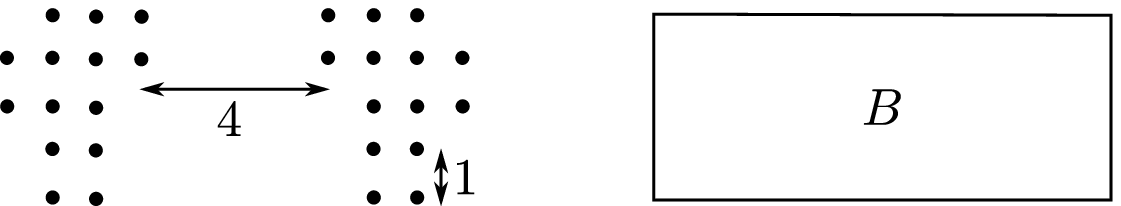}}
\caption{A configuration-$\mathbf{B}$ (left), and its schematic abbreviation (right).\label{fig:configuration_B}}
\end{center}
\end{figure}
\else
\begin{figure}[t]
\begin{center}
\scalebox{0.7}{\includegraphics{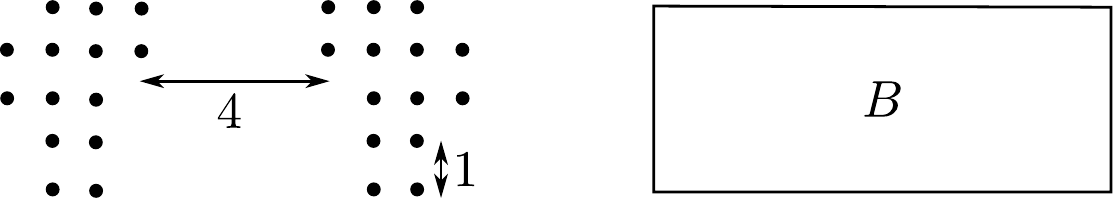}}
\caption{A configuration-$\mathbf{B}$ (left), and its schematic abbreviation (right).\label{fig:configuration_B}}
\end{center}
\end{figure}
\fi

Let $(P, {\cal H})$ be a $\theta$-ribbon from $x$ to $y$, with ${\cal H}=\{H_i\}_{i=1}^k$, and let $Y$ be the corresponding $2$-chain.
Suppose that there exists some \emph{odd} $j\in [k-15]$, such that for any $r,r'\in \{j,\ldots,j+14\}$, we have $H_r=H_{r'}$.
Let $Z$ be a configuration-$\mathbf{B}$ contained in the $2$-plane $H_j$.
%Suppose further that $Z \cup Y$ appear in $H_j$ as in Figure \ref{fig:neighbor_H}.
Suppose that we replace a subset of $Y$ by a  configuration-$\mathbf{B}$ $Z$, such that $Z$ is contained in $H_j$, and is as in figure \ref{fig:attaching_B}.
Then, we say that resulting point-set $Y'$ is obtained by \emph{attaching} $Z$ to $Y$ at $j$.

\ifabstract
\begin{figure}[t]
\begin{center}
\scalebox{0.45}{\includegraphics{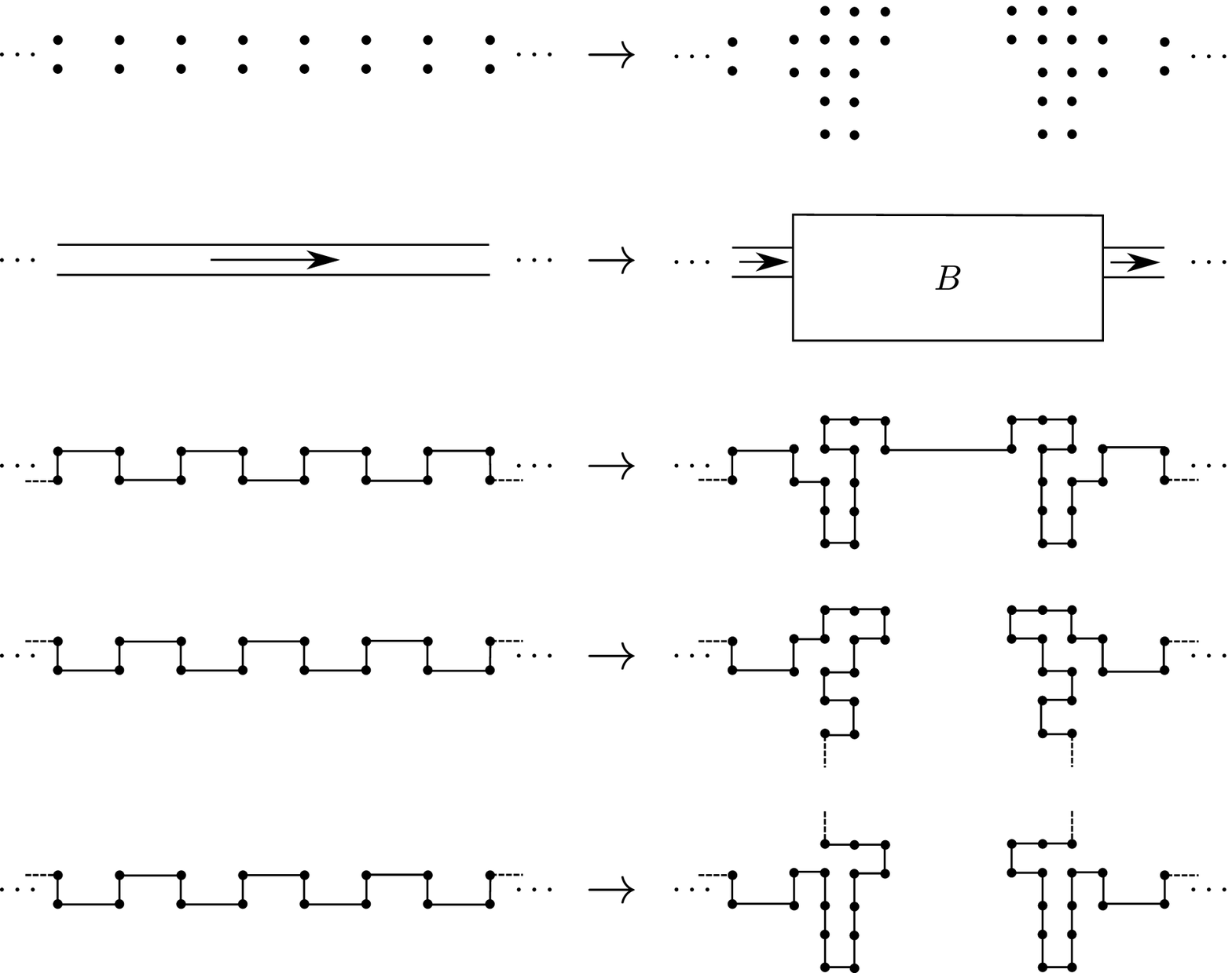}}
\caption{From top to bottom: Attaching a configuration-$\mathbf{B}$ to a 2-chain, its schematic abbreviation, the optimal path for mode 2, and the two possible optimal paths for mode 1.\label{fig:attaching_B}}
\end{center}
\end{figure}
\else
\begin{figure}[t]
\begin{center}
\scalebox{0.7}{\includegraphics{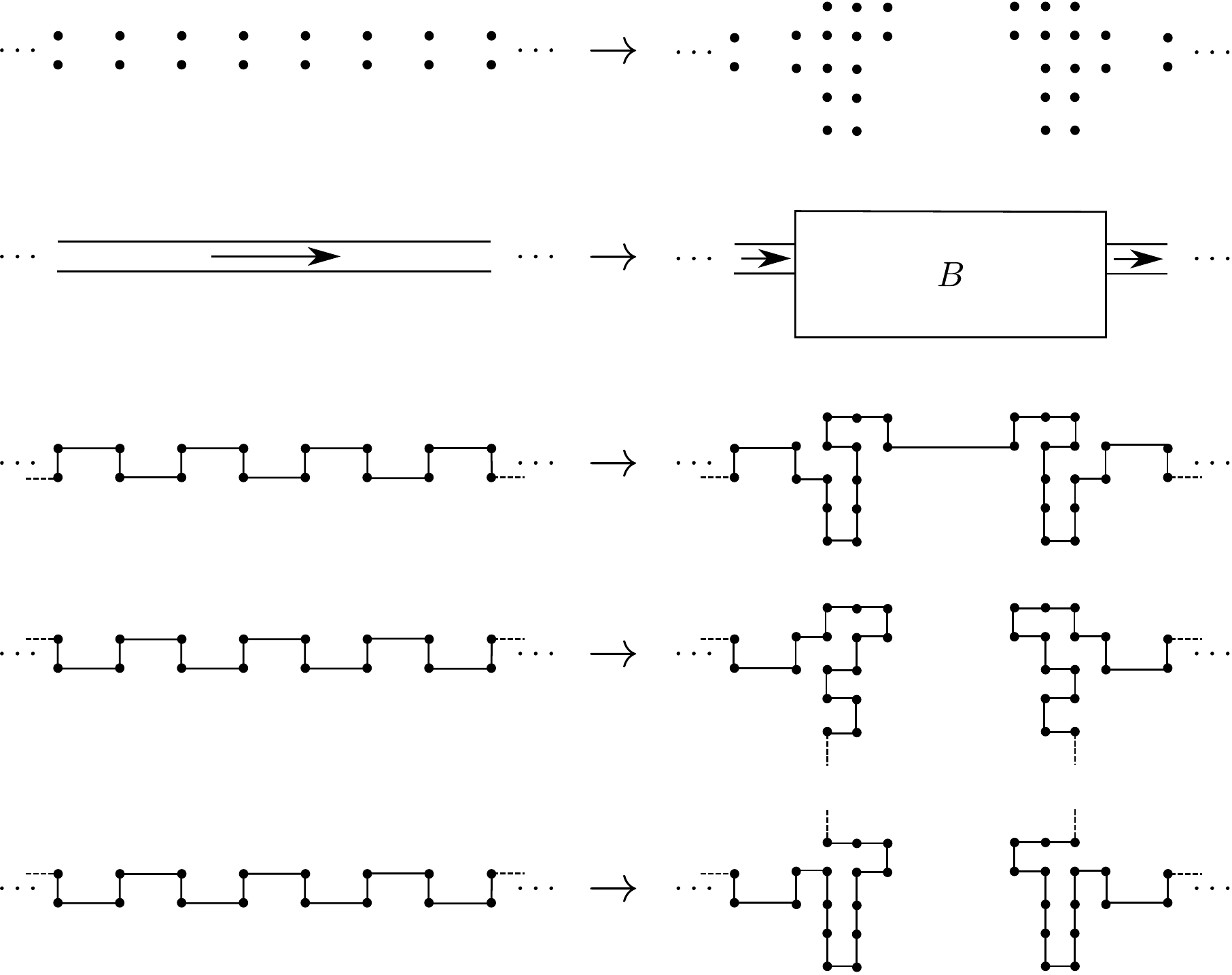}}
\caption{From top to bottom: Attaching a configuration-$\mathbf{B}$ to a 2-chain, its schematic abbreviation, the optimal path for mode 2, and the two possible optimal paths for mode 1.\label{fig:attaching_B}}
\end{center}
\end{figure}
\fi

\textbf{The reduction.}
Let $(V,D,C)$ be an instance of a binary constraint satisfaction problem, with primal graph $G=\Rs[n,d]$.
We may assume w.l.o.g.~that every constraint is binary (i.e.~there are no unary constraints).

First, we encode the variables.
Let $\gamma = \gamma(D)$ be a parameter to be determined later.
For any $\mathbf{r} = (r_1,\ldots,r_d) \in \mathbb{Z}^d$, let 
\ifabstract
$U(\mathbf{r}) := \{(x_1,\ldots,x_d) \in \mathbb{R}^d : \text{ for all } i\in [d], r_i \gamma \leq x_i < (r_i+1) \gamma\}$.
\else
\[
U(\mathbf{r}) := \{(x_1,\ldots,x_d) \in \mathbb{R}^d : \text{ for all } i\in [d], r_i \gamma \leq x_i < (r_i+1) \gamma\}.
\]
\fi
Let us identify $V$ with $[n]^d$, in the obvious way.
For each $\uu\in V$, we introduce a family of $d$ 2-chains $\Gamma(\uu,1), \ldots, \Gamma(\uu,|D|)$.
We will ensure that
\ifabstract
$\bigcup_{i=1}^{|D|} \Gamma(\uu,i) \subset U(\uu) \cup \bigcup_{i=1}^d \left(U(\uu-\mathbf{e_i}) \cup U(\uu+\mathbf{e_i}) \right)$,
\else
$\bigcup_{i=1}^{|D|} \Gamma(\uu,i) \subset U(\uu) \cup \bigcup_{i=1}^d \left(U(\uu-\mathbf{e_i}) \cup U(\uu+\mathbf{e_i}) \right)$,
\fi
where $\mathbf{e}_1,\ldots,\mathbf{e}_d$ is the standard orthonormal basis in $\mathbb{R}^d$.
We need to enforce that in any optimal solution, exactly one of the 2-chains $\Gamma(\uu,1)$, $\ldots$, $\Gamma(\uu,|D|)$ is traversed in mode 2.
Intuitively, this will correspond to assigning the value $i$ to variable $\uu$.
To that end, we construct the 2-chains $\Gamma(\uu,i)$ such that there exists a 2-plane $h(\uu,i)$, where subsets of the 2-chains are arranged as in Figure \ref{fig:var}.
Namely, for any $i\in [|D|]$ we introduce a configuration-$\mathbf{B}$ $B(\uu,i)$ that is attached to $\Gamma(\uu,i)$ at some $j_i$.
Moreover, for any $i\in [|D|-1]$, we introduce a configuration-$\mathbf{H}$ $H(\uu, i)$, such that $H(\uu,i)$ is the right neighbor of $\Gamma(\uu,i)$ at $j_i$, and the left neighbor of $\Gamma(\uu,i+1)$ at $j_{i+1}$.

\ifabstract
\begin{figure}[t]
\begin{center}
\scalebox{0.45}{\includegraphics{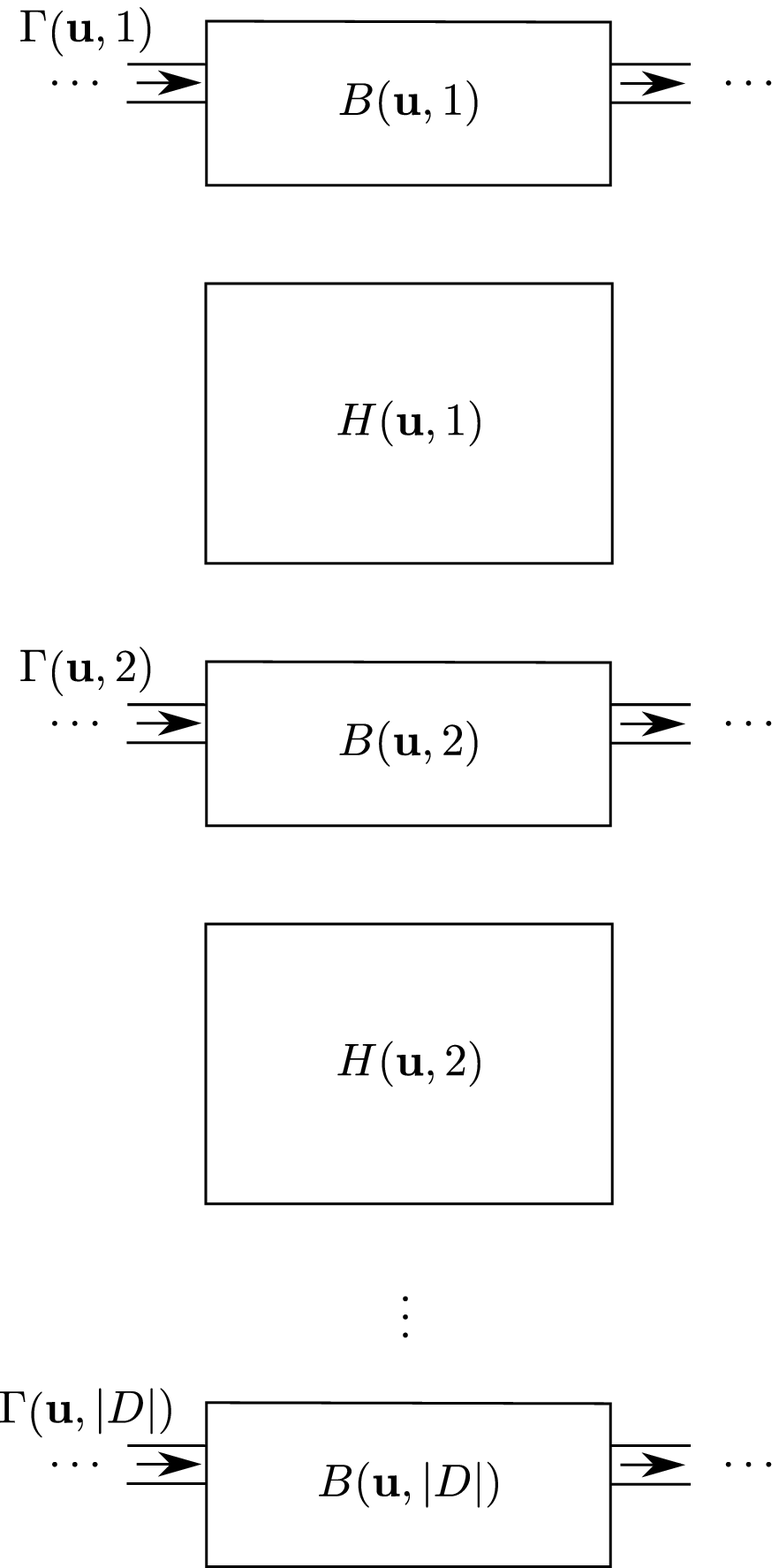}}
\caption{Gadget corresponding to some variable $u$.\label{fig:var}}
\end{center}
\end{figure}
\else
\begin{figure}[t]
\begin{center}
\scalebox{0.6}{\includegraphics{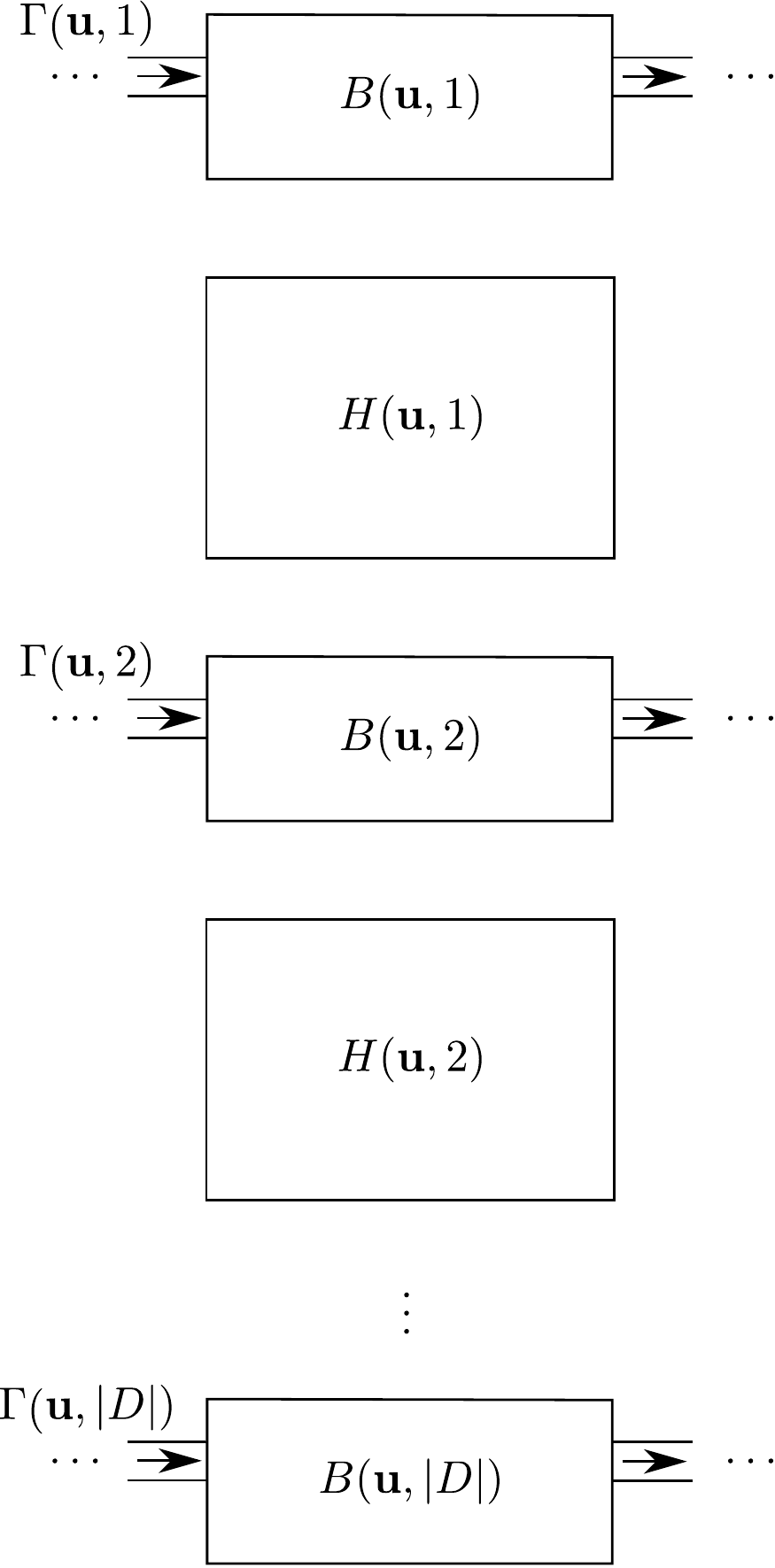}}
\caption{A gadget corresponding to some variable $u$.\label{fig:var}}
\end{center}
\end{figure}
\fi

Next, we encode the constraints.
Let $\langle (\uu,\vv), R \rangle \in C$ be a constraint.
For every pair of values $(i,j)\in D^2$, with $(i,j)\notin R$, we need to ensure that in any optimal solution, at most one of the 2-chains $\Gamma(\uu,i)$, $\Gamma(\vv,j)$ is traversed in mode 1.
To that end, we add two new vertices $x(\uu,i,\vv,j)$, $y(\uu,i,\vv,j)$, and a new 2-chain $\Gamma(\uu,i,\vv,j)$ from $x(\uu,i,\vv,j)$ to $y(\uu,i,\vv,j)$.
We arrange the 2-chains such that there exists a 2-plane $h(\uu,\vv,i,j)$, with subsets of the 2-chains $\Gamma(\uu,i)$, and $\Gamma(\vv,j)$ being arranged in $h(\uu,\vv,i,j)$ as in Figure \ref{fig:constraint}.
Namely, we introduce configurations-$\mathbf{B}$ $B(\uu,i,\vv,j,1)$, $B(\uu,i,\vv,j,2)$, $B(\uu,i,\vv,j,3)$, such that $B(\uu,i,\vv,j,1)$ is attached to $\Gamma(\uu,i)$ at some $\ell_i$, $B(\uu,i,\vv,j,2)$ is attached to $\Gamma(\uu,i,\vv,j)$ at some $\ell_2$, and $B(\uu,i,\vv,j,3)$ is attached to $\Gamma(\vv,j)$ at some $\ell_3$.
Moreover, we add a configuration-$\mathbf{H}$ $H(\uu,i,\vv,j,1)$ that is a right neighbor of $\Gamma(\uu,i)$ at $\ell_1$, and a left neighbor of $\Gamma(\uu,i,\vv,j)$ at $\ell_2$, and a configuration-$\mathbf{H}$ $H(\uu,i,\vv,j,2)$ that is a right neighbor of $\Gamma(\uu,i,\vv,j)$ at $\ell_2$, and a left neighbor of $\Gamma(\vv,j)$ at $\ell_3$.

\ifabstract
\begin{figure}[t]
\begin{center}
\scalebox{0.45}{\includegraphics{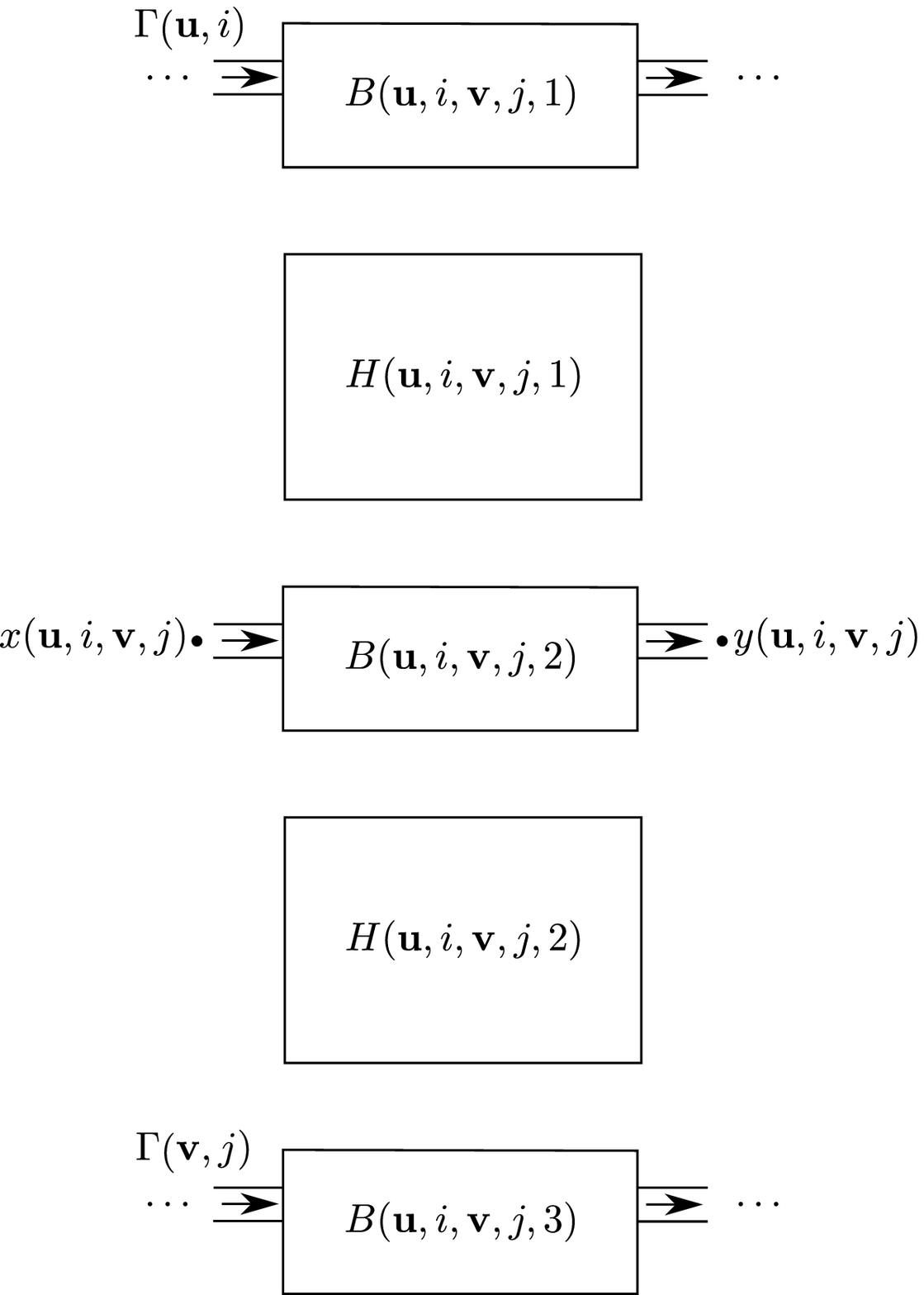}}
\caption{Part of a gadget encoding $(i,j)\notin R$, for some constraint $\langle (\uu,\vv),R\rangle$.\label{fig:constraint}}
\end{center}
\end{figure}
\else
\begin{figure}[t]
\begin{center}
\scalebox{0.6}{\includegraphics{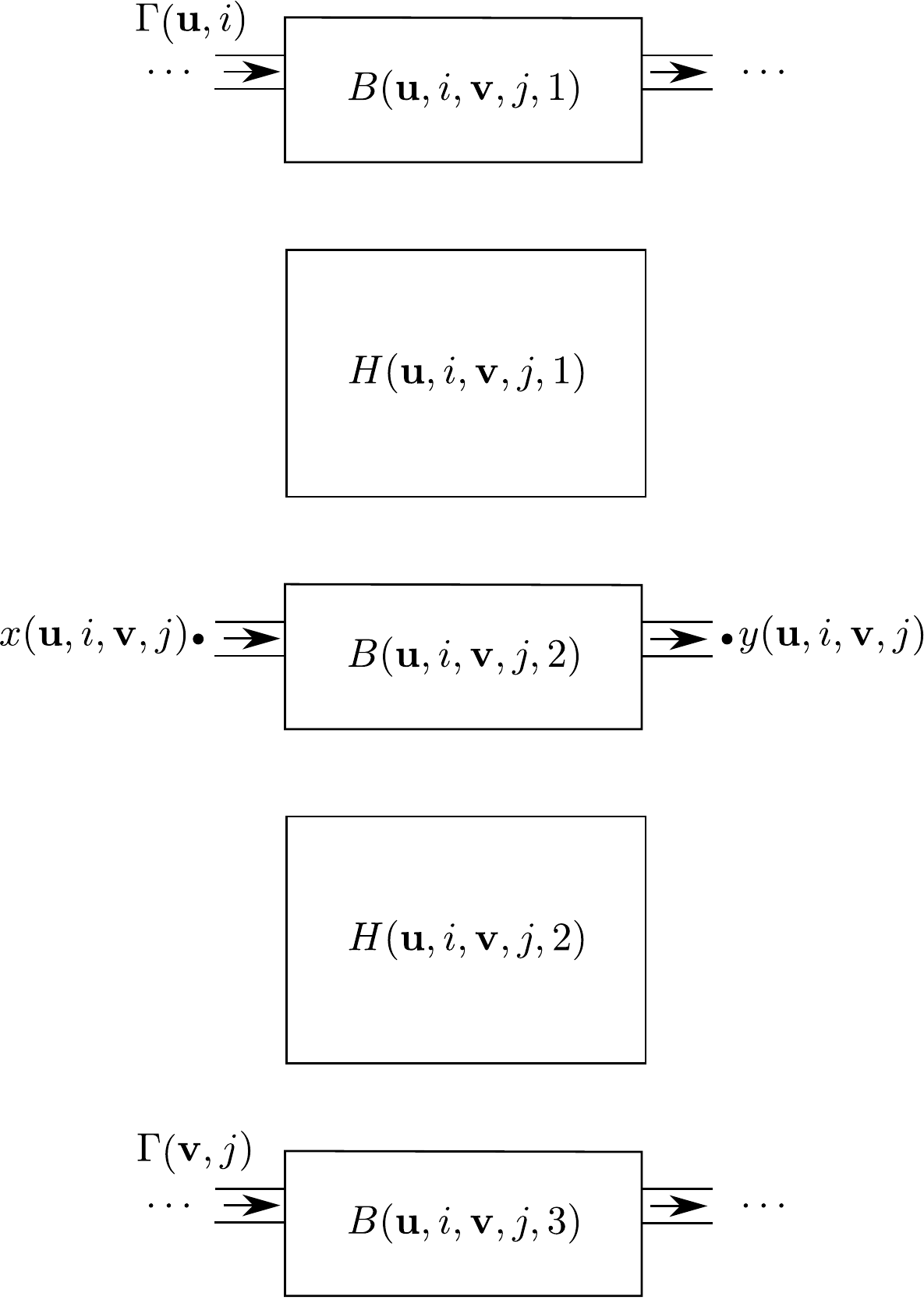}}
\caption{Part of a gadget encoding $(i,j)\notin R$, for some constraint $\langle (\uu,\vv),R\rangle$.\label{fig:constraint}}
\end{center}
\end{figure}
\fi

Finally, we need to ensure that the optimal solution induces a traversal of all the 2-chains, such that each 2-chain is traversed without interruptions.
This can be done by introducing 1-chains between the endpoints of 2-chains which we want to appear consecutively in the optimal traversal.
Initially, we unmark all 2-chains in the construction.
Observe that the graph $G$ is Hamiltonian.
Fix a Hamiltonian path $P$ in $G$.
We construct a total ordering of all the 2-chains in the construction as follows.
We start from the empty ordering, and we consider all vertices in the order they are visited by $P$.
When considering a vertex $u$, we extend the ordering by appending all the 2-chains $\Gamma(\uu,i)$, for all $i\in [|D|]$.
Next, we also append all the 2-chains $\Gamma(\uu,i,\vv,j)$, for all $i,j\in [|D|]$, and $v\in V(G)$, that we have not added to the ordering yet.
This process  clearly results in a total ordering of all the 2-chains in the construction.
Let $k$ be the total number of 2-chains, for any $i\in [k]$, let $p_i$, $q_i$ be points such that the $i$-th chain is from $p_i$ to $q_i$.
For any $i\in [k-1]$, we add a 1-chain from $q_i$ to 
$q_{i+1}$.
When adding a 1-chain $Y$ between the endpoints of two 2-chains involving only one variable $\uu$ (i.e.~$\Gamma(\uu,i)$, and $\Gamma(\uu,i+1)$ for some $i$), we ensure that $Y\subset U(\uu)$.
Similarly, when we add a 1-chain $Y$ between the endpoints of two 2-chains involving only two variable $\uu,\vv$ (i.e.~either $\Gamma(\uu,i)$ and $\Gamma(\uu,j,\vv,\ell)$ for some $i,j,\ell$, or $\Gamma(\uu,i)$ and $\Gamma(\vv,j)$ for some $i,j$), we ensure that $Y\subset U(\uu) \cup U(\vv)$.

Finally, we introduce two points $p^*$, and $q^*$, and we add a 1-chain from $p^*$ to $p_1$, and a 1-chain from $q_t$ to $q^*$.
We choose the point $p^*$ to be in $U(-2 \cdot \mathbf{e}_1)$,
and the point $q^*$ to be in $U((n+2)\cdot \mathbf{e}_1)$.
We can clearly chose the 1-chains from $p^*$ to $p_1$, and from $q_t$ to $q^*$ so that the following is satisfied.

\begin{lemma}\label{lem:endpoints}
Any optimal Traveling Salesperson path from in the constructed instance has endpoints $p^*$ and $q^*$.
\end{lemma}
\fi

\ifmove
\textbf{Analysis.}
Let $X \subset \mathbb{R}^d$.
For some integer $b$, a subset $Y \subseteq X$ is called a \emph{$b$-component} if for all $c\in Y$, we have 
\[
\min\{\|c-c'\|_2 : c'\notin Y\} \geq b,
\]
and
\[
\max\{\|c-c'\|_2 : c' \in Y\} < b,
\]
and $Y$ is maximal with respect to this property.
Note that $X$ might not contain any $b$-components.

For some integer $k$, a \emph{$k$-path} in $X$ is a set of $k$ vertex-disjoint, not closed paths covering $X$.
In particular, Traveling Salesperson paths are $1$-paths.
We say that a subset $Y \subseteq X$ is \emph{$(b,k_0)$-compact}, if for all positive integers $k\leq k_0$, an optimal $k$-path has length less than $b$ plus the length of the optimal $(k+1)$-path.
Note that any $b$-component is trivially $(b,|G|-1)$-compact.

We recall the following from \cite{DBLP:journals/tcs/Papadimitriou77}.
For completeness, we reproduce its proof here.

\begin{lemma}[Papadimitriou \cite{DBLP:journals/tcs/Papadimitriou77}]\label{lem:pap}
Let $d\geq 2$, and let $X \subset \mathbb{R}^d$ be a finite set of points.
Suppose that we have $N$ $a$-components $G_1,\ldots,G_N$ in $X$, such that the distance between any two components is at least $2a$, and $G_0$, the remaining part of $X$, is $(a,N+1)$-compact.
Suppose that any optimal Traveling Salesperson path of $X$ has its endpoints on $G_0$, and that it does not contain edges between any two $a$-components of $X$.
Let $L_1,\ldots,L_N$ be the lengths of the optimal $1$-paths of $G_1,\ldots,G_N$ and $L_0$ the length of the optimal $(N+1)$-path of $G_0$.
If there is a $1$-path $P$ of $X$ consisting of the union of an optimal $(N+1)$-path of $G_0$, $N$ optimal $1$-paths of $G_1,\ldots,G_N$, and $2N$ edges of length $a$ connecting $a$-components to $G_0$, then $P$ is optimal.
If no such $1$-path exists, the optimal $1$-path of $X$ has length greater than $L=L_0 + L_1 + \ldots + L_N + 2Na$.
\end{lemma}
\begin{proof}
Let $P'$ be an optimal path for $X$.
Suppose that for any $i\in \{0,\ldots,N\}$, the restriction $P_i'$ of $P'$ on $G_i$ is a $k_i$-path, for some integer $k_i\geq 1$.
Since any optimal path for $X$ must have both endpoints in $X_0$, it follows that $k_0=\sum_{i=1}^N k_i + 1$.
Let $L'$ be the length of $P'$, and for each $i\in \{0,\ldots,N\}$, let $L_i'$ be the length of the optimal $k_i$-path for $X_i$.
We have $L_0'\geq L_0 - (k_0-N)a$, and for any $i\in [N]$, we have $L_i'\geq L_i-(k_i-1)a$.
It follows that $L'\geq \sum_{i=1}^N (L_i-(k_i-1)a) + 2(k_0-1)a + L_0 - (k_0-N)a$.
Since $P'$ is optimal, we have $L'=L$, which implies that for any $i\in \{1,\ldots,N\}$, we have $k_i=1$, concluding the proof.
\end{proof}
\fi
\ifmain
\begin{lemma}\appstar\label{lem:tsp_reduction}
Let $d \geq 2$.
Let $\phi=(V,D,C)$ be an instance of a constraint satisfaction problem with domain size $|D|=\delta$, with constraint graph $G=\Rs[n,d]$.
Then, there exists a polynomially-time computable instance $(X,\alpha)$ of TSP in $d$-dimensional Euclidean space, with $|X| \leq n^d \cdot |D|^{O(1)}$, such that the length of the shortest TSP tour for $X$ is at most $\alpha$, if and only if $\phi$ is satisfiable.
\end{lemma}
\fi
\ifmove
\ifappendix
\begin{proof}[Proof (of Lemma~\ref{lem:tsp_reduction})]
\else
\begin{proof}
\fi
Given $\phi$, we compute the instance $(X,\alpha)$ as described above.
Let $G_1,\ldots,G_N$ be all those components-$\mathbf{H}$ in the constructed instance, and let $G_0 = X \setminus \bigcup_{i=1}^N G_i$.

Let $a=20$.
By leaving enough space between the different components in the above construction, we can ensure that the following properties are satisfied:
\begin{description}
\item{(P1)}
Every configuration-$\mathbf{H}$ used in the construction is an $a$-component.
\item{(P2)}
The distance between any two distinct configurations-$\mathbf{H}$ is at least $2a$.
\item{(P3)}
Every optimal Traveling Salesperson path of $X$ has its endpoints in $G_0$; this follows by Lemma \ref{lem:endpoints} since $p^*,q^*\in G_0$.
\item{(P4)}
The set $G_0$ is $(a,N+1)$-compact.
We can satisfy this property by making sure that the following conditions hold:
(i) Distinct 2-chains are at distance at least 20 from each other.
(ii) Distinct 1-chains are at distance at least 20 from each other. (iii) For any 1-chain $C$ and 2-chain $C'$ that do not share common endpoints, we have that $C$ and $C'$ are at distance at least 20 from each other.
(iv) Let $C$ be a 1-chain that shares an endpoint $p$ with some 2-chain $C'$.
Let $(Q,{\cal H})$ be the ribbon corresponding to $C'$.
Then, there exists a line $\ell$ in $\mathbb{R}^d$, with $p\in \ell$, such that the $20$ segments in $Q$ closest to $q$ lie in $\ell$, and the $20$ points in $C$ closest to $p$ also lie in $\ell$; all other points in $C$ are at distance at least 20 from $C'$.
\end{description}

It is easy to check that the above conditions can be satisfied by setting the width $\gamma$ of each $U(\uu)$ to be a sufficiently large polynomial in $|D|$, say $\gamma = \Theta(|D|^5)$.

Suppose now that $\phi$ is satisfiable, and let $T:V(G) \to \{1,\ldots,|D|\}$ be a satisfying assignment for $\phi$.
We can build a Traveling Salesperson path for $X$ as follows.
We start at $p^*$, and we traverse the first 1-chain originating at $p^*$.
When we reach a 2-chain $C$, we distinguish between the following cases:

(i) Suppose first that $C$ is a 2-chain of the form $\Gamma(\uu, i)$ for some $\uu \in V(G)$, and $i\in [|D|]$.
If $T(\uu)=i$, then we begin traversing $C$ in mode 2, and if $T(\uu)\neq i$, then we begin traversing $C$ in mode 1.
When we reach the configuration-$\mathbf{B}$ $B(\uu,i)$, if $C$ is traversed in mode 2, then we traverse $B(\uu,i)$ as in figure \ref{fig:configuration_B}; if $C$ is traversed in mode 1, then if $T(\uu)>i$, then we traverse $B(\uu,i)$ and the configuration-$\mathbf{H}$ $H(\uu,i)$; similarly, if $T(\uu)<i$, then we traverse $B(\uu,i)$ and the configuration-$\mathbf{H}$ $H(\uu,-1)$ (see Figure \ref{fig:configuration_B}).
Similarly, if we reach a configuration-$\mathbf{B}$ $B(\uu,i,\vv,j,1)$, if $C$ is traversed in mode 1 then we traverse $B(\uu,i,\vv,j,1)$ together with $H(\uu,i,\vv,j,1)$;
if we reach a configuration-$\mathbf{B}$ $B(\vv,i,\uu,j,3)$, if $C$ is traversed in mode 1 then we traverse $B(\vv,i,\uu,j,3)$ together with $H(\vv,i,\uu,j,2)$ 

(ii)
Suppose next that $C$ is a 2-chain of the form $\Gamma(\uu, i, \vv, j)$ for some $\uu,\vv \in V(G)$, and $i,j\in [|D|]$.
This implies that there exists a constraint $\langle (\uu,\vv), R\rangle$, with $(i,j)\notin R$.
Since $T$ satisfies $\phi$, it follows that either $T(\uu)\neq i$, or $T(\vv)\neq j$.
We have the following sub-cases:
(ii.1) If $T(\uu)\neq i$, and $T(\vv)\neq j$, then we traverse $\Gamma(\uu,i,\vv,j)$ in mode 2.
(ii.2) If $T(\uu)=i$, then we traverse $\Gamma(\uu,i,\vv,j)$ in mode 1, and we traverse $B(\uu,i,\vv,j,2)$ together with $H(\uu,i,\vv,j,1)$.
(ii.3) If $T(\vv)=j$, then we traverse $\Gamma(\uu,i,\vv,j)$ in mode 1, and we traverse $B(\uu,i,\vv,j,2)$ together with $H(\uu,i,\vv,j,2)$.

Upon reaching the end of a 2-chain, we traverse the 1-chain that is attached to it, and we proceed to the next 2-chain.
Eventually, we reach $q^*$, and the path terminates.
It is immediate to check that the resulting path $P$ visits all points in $X$.

By Lemma \ref{lem:configuration_H} we have that every configuration-$\mathbf{H}$ $G_i$ is visited in an optimal way.
By Lemma \ref{lem:two_chain} we have that traversing a 2-chain in either mode 1, or mode 2, results in an optimal path.
Moreover, it follows by \cite{DBLP:journals/tcs/Papadimitriou77} that traversing each configuration-$\mathbf{B}$ as in Figure \ref{fig:configuration_B} results in an optimal $(N+1)$-path for $G_0$.
Therefore, we have the following property:
\begin{description}
\item{(P5)}
The path $P$ consists of the union of an optimal $(N+1)$-path for $G_0$, $N$ optimal $1$-paths for $G_1,\ldots,G_N$, and $2N$ edges of length $a$ connecting $a$-components to $G_0$.
\end{description}
Combining properties (P1)--(P5) with Lemma \ref{lem:pap} we deduce that $P$ is optimal.

On the other hand, any optimal path $P$ must consist of an optimal $(N+1)$-path for $G_0$, $N$ optimal $1$-paths for $G_1,\ldots,G_N$, and $2N$ edges of length $a$ connecting $G_0$ with $G_1,\ldots,G_N$ (two edges for each $a$-component).
By Lemma \ref{lem:configuration_H} we have that for each $i\in \{1,\ldots,N\}$, the configuration-$\mathbf{H}$ $G_i$ is visited as in Figure \ref{fig:configuration_H}.
This implies that $P$ must contain two edges between each configuration-$\mathbf{H}$ $G_i$ and some configuration-$\mathbf{B}$ that is attached to some 2-chain $C$, such that $G_i$ is a neighbor of $C$, and $C$ is traversed in mode 1.
This implies by the construction that for each $\uu\in V(G)$, exactly one of the 2-chains $\Gamma(\uu,j)$ is traversed in mode 2.
We can construct an assignment $T$ for $\phi$ by setting $T(\uu)=j$, where $\Gamma(\uu,j)$ is traversed in mode 2.
It remains to check that $T$ satisfies $\phi$.
Let $\langle(\uu,\vv),R\rangle$ be a constraint, and let $(i,j)\notin R$.
Again, by the construction, we have that at least one of the 2-chains $\Gamma(\uu,i)$, $\Gamma(\vv,j)$ is traversed in mode 2, and therefore either $T(\uu)\neq i$, or $T(\vv)\neq j$.
It follows that $T$ satisfies $\phi$.

It remains to bound $|X|$.
Observe that there are $O(n^d \cdot |D|)$ configurations~-$\mathbf{H}$, and configurations-$\mathbf{B}$, and each one has a constant number of points.
There are also $O(n^d)$ 1-chains, and 2-chains.
Each such chain intersects at most a constant number of sets $U(\uu)$, for $\uu\in \mathbb{Z}^d$.
Therefore, the total number of points in each such chain is at most $O(\gamma) = O(|D|^5)$.
It follows that $|X| = O(n^d \cdot |D|^6)$.
\end{proof}
\fi

\ifmain
\begin{proof}[\ifabstract (of Theorem \ref{th:intro-tsp})\else Proof of Theorem \ref{th:intro-tsp}\fi]
It follows by Theorem \ref{th:gridnolog} \& Lemma \ref{lem:tsp_reduction}.
\end{proof}

Recall that cycle-TSP is the variant of TSP where one seeks to find a cycle visiting all points.
We can prove the same lower bound for cycle-TSP, using a simple modification of the above reduction.
We remark that the same modification was used in \cite{DBLP:journals/tcs/Papadimitriou77} to show that cycle-TSP in the Euclidean plane is NP-complete.

\ifabstract
\begin{theorem}\label{th:cycle-tsp}
If for some $d\ge 2$ and $\epsilon>0$, cycle-{TSP} in $d$-dimensional Euclidean space can be solved in time\\ $2^{O(n^{1-1/d-\epsilon})}$, then ETH fails.
\end{theorem}
\else
\begin{theorem}\label{th:cycle-tsp}
If for some $d\ge 2$ and $\epsilon>0$, cycle-{TSP} in $d$-dimensional Euclidean space can be solved in time $2^{O(n^{1-1/d-\epsilon})}$, then ETH fails.
\end{theorem}
\fi
\begin{proof}
We use the same reduction as for the case of path-TSP above.
The only modification needed is to connect $p^*$ with $q^*$ via a 1-chain, that is at distance at least 20, say, from all other gadgets used in the reduction.
\end{proof}
\fi
\iffalse
\begin{remark}
Precision of coordinates.
\textbf{TODO}
\end{remark}
\fi

%%% Local Variables: 
%%% mode: latex
%%% TeX-master: "main"
%%% End: 

\ifabstract\clearpage\fi
\bibliographystyle{abbrv}
\bibliography{d-dim-geo}

\appendix
\appendixtrue\ifabstract\movetrue
\mainfalse

\fi

\section{An alternative proof of Theorem \ref{thm:embedding}}\label{sec:embedding_alternative}

We give an alternative proof of Theorem \ref{thm:embedding}, based on the expansion properties of the graph $\Hs[n,d]$, and the multi-commodity max-concurrent-flow/sparsest-cut duality.

We begin with some definitions.
Let $G$ be a graph, and let $(\capa, \dem)$ be a multi-commodity flow instance on $G$, i.e.~with $\capa:E(G)\to \mathbb{R}_{\geq 0}$, and $\dem:V(G) \times V(G) \to \mathbb{R}_{\geq 0}$.
For a subset $S \subset V(G)$, with $S\neq \emptyset$, the \emph{sparsity} of $S$ (w.r.to~$\mu$) is defined to be
\[
\phi(S; \mu) := \frac{\sum_{\{u,v\}\in E(G):u\in S, v\notin S}\capa(\{u,v\})}{\sum_{u \in S, v\notin S} \dem(u,v)}.
\]
We also define 
\[
\phi(G; \mu) := \min_{S\subset V(G): S \neq \emptyset} { \phi(S) }.
\]
When every edge $e\in E(G)$ has unit capacity, i.e.~$\capa(e)=1$, and there is a unit demand between every pair of distinct vertices, i.e. $\dem(u,v)=1$ for all $u\neq v\in V(G)$, we use the notation $\phi(S;\mu)=\phi(S)$, and $\phi(G;\mu)=\phi(G)$.

The \emph{maximum concurrent flow} of a multi-commodity flow instance $\mu=(\capa,\dem)$ supported on a graph $G$, denoted by $\maxflow(G; m)$, is defined to be the supremum $\eps>0$, such that the multi-commodity flow instance $(\capa, \eps\cdot \dem)$ admits a routing in $G$ with congestion at most 1 (i.e.~without violating the capacity constrains).

\begin{theorem}[Linial, London \& Rabinovich \cite{DBLP:journals/combinatorica/LinialLR95}]\label{thm:LLR}
For any $n$-vertex graph $G$, and for any multi-commodity flow instance $\mu$ on $G$, we have
$\maxflow(G; \mu) = \Omega(1/\log n) \cdot \phi(G; \mu)$.
\end{theorem}

For an $n$-vertex graph $G$, its adjacency matrix $A_G \in \mathbb{R}^{n\times n}$ is defined by $A_G = \{a_{i,j}\}_{i,j}$, with $a_{i,j} = 1$ if $\{i,j\} \in E(G)$, and $a_{i,j}=0$ otherwise, were we have identified $V(G)$ with the set $\{1,\ldots,n\}$.
Let also $D_G \in \mathbb{R}^{n \times n}$ be the diagonal matrix defined by $D_G=\{d_{i,j}\}_{i,j}$, where $d_{i,i} = \deg_G(i)$.
The \emph{Laplacian} matrix of $G$, denoted by $L_G$ is defined to be
\[
L_G := D_G - A_G.
\]
For any $i\in \{1,\ldots,n\}$, let $\lambda_i(G)$ denote the $i$-th smallest eigenvalue of $L_G$.
For any graph $G$, we have
\begin{align}
\phi(G) &\geq \frac{\lambda_2(G)}{2\cdot |V(G)|} \label{eq:phi_lambda2}
\end{align}

For a pair of graphs $G$, $G'$, with the \emph{Cartesian product} of $G$ and $G'$ is defined to be the graph $\Gamma$, denoted by $G \times G'$, with 
$V(\Gamma) = V(G) \times V(G')$,
and
$\{(u,u'), (v,v')\} \in E(G)$ if and only if either
$u=v$, and $\{u',v'\}\in E(G')$, or $\{u,v\}\in E(G)$, and $u'=v'$.

\begin{lemma}[Fiedler \cite{Fiedler1973}]\label{lem:eig_interleaving}
For any pair of graphs $G$, $G'$, with $|V(G)|=n$, $|V(G')|=n'$, we have
\[
\bigcup_{i=1}^{n \cdot n'}\{ \lambda_i(G \times G')\} = \bigcup_{i=1}^{n} \bigcup_{j=1}^{n'} \{\lambda_i(G) + \lambda_j(G')\}.
\]
\end{lemma}

\begin{lemma}[Fiedler \cite{Fiedler1973}]\label{lem:lambda2_clique}
For any $n>0$, we have $\lambda_2(K_n) = n$.
\end{lemma}
 
We are now ready to give an alternative proof of Theorem 2.15, with a slightly worse bound on the depth of the resulting embedding.

\ifabstract
\begin{proof}[Alternative proof of Theorem \ref{thm:embedding}]
\else
\begin{proof}[Alternative proof of Theorem \ref{thm:embedding} with slightly worse depth]
\fi
We observe that the graph $\Hs[n,d]$ is the $d$-wise Cartesian product of $K_k$, i.e.
\[
\Hs[n,d] = \underbrace{K_n \times \ldots \times K_n}_{d \text{ times}}
\]
For any graph $K$, we have $\lambda_1(K) = 0$.
Therefore, by Lemma \ref{lem:eig_interleaving} it follows that for any pair of graphs $J$, $J'$, we have
\begin{align}
\lambda_2(J\times J') &= \min\{ \lambda_2(J), \lambda_2(J') \} \label{eq:lambda2_product}
\end{align}
Combining \eqref{eq:lambda2_product} with Lemma \ref{lem:lambda2_clique}, we obtain
\begin{align}
\lambda_2(\Hs[n,d]) &= \lambda_2(K_n) = n. \label{eq:H_lambda2}
\end{align}
By \eqref{eq:H_lambda2} \& \eqref{eq:phi_lambda2} we get
\begin{align}
\phi(\Hs[n,d]) &\geq \frac{n}{2m}. \label{eq:H_phi}
\end{align}

%Since $G$ is connected, we have $|V(G)| < 2m$.
Let $f:V(G) \to [n]^d$ be an arbitrary map that sends at most $2$ vertices in $V(G)$ to any vertex in $[n]^d$.
Consider a multi-commodity flow instance 
$\mu = (\capa, \dem)$
on $\Hs[n,d]$, where every edge $e \in E(\Hs[n,d])$ has capacity $\capa(e)=1$, 
%and for every edge $\{u,v\} \in E(G)$, we have $\dem(f(u),f(v))=1$.
and for every $u,v\in [n]^d$, we set
\[
\dem(u,v) = |\{\{u',v'\}\in E(G) : f(u')=u \text{, and } f(v')=v \}|.
\]

For any $A\subseteq [n]^d$, and for any $u\in [n]^d$, let us write
\[
\dem(u,A)=\sum_{v\in A} \dem(u,v).
\]
Since $G$ has maximum degree $\dmax$,
and for every $u\in V(G)$, $|f^{-1}(u)| \leq 2$,
it follows that the total demand incident to any vertex of $\Hs[n,d]$ is at most $2 \dmax$, i.e.~for any $u\in [n]^d$, we have $\dem(u,[n]^d) \leq 2\dmax$.
Therefore, for any $S\subset [n]^d$, with $S\neq \emptyset$, we have
\begin{align}
\sum_{u\in S, v\notin S} \dem(u,v) &= \min\left\{ \sum_{u\in S} \dem(u,[n]^d)\setminus S), \sum_{u\in [n]^d) \setminus S} \dem(u, S) \right\} \notag \\
 &\leq  \min\left\{ \sum_{u\in S} \dem(u,[n]^d), \sum_{u\in [n]^d) \setminus S} \dem(u, [n]^d) \right\} \notag \\
 &\leq 2\dmax \cdot \min\{|S|, |[n]^d\setminus S|\} \notag \\
 &\leq \frac{2 \dmax}{m} \cdot |S| \cdot |[n]^d \setminus S| \label{eq:sum_d}
\end{align}
Combining \eqref{eq:sum_d} \& \eqref{eq:H_phi} we obtain
%\begin{align}
%\phi(\Hs[n,d]; \mu) &\geq \frac{m}{2 \dmax} \phi(\Hs[n,d]) \geq \frac{n}{4\dmax} \label{eq:sparsity}
%\end{align}
\begin{align}
\phi(\Hs[n,d]; \mu) &= \min_{S\subset [n]^d: S \neq \emptyset} \frac{\sum_{\{u,v\}\in E(\Hs[n,d]):u\in S, v\notin S}\capa(\{u,v\})}{\sum_{u \in S, v\notin S} \dem(u,v)}\notag\\
&\geq \frac{m}{2\dmax} \min_{S\subset [n]^d: S \neq \emptyset} \frac{|\{\{u,v\}\in E(\Hs[n,d]):u\in S, v\notin S|}{|S| \cdot |[n]^d \setminus S|}\notag\\
&= \frac{m}{2 \dmax} \phi(\Hs[n,d]) \notag \\
&\geq \frac{n}{4\dmax} \label{eq:sparsity}
\end{align}
By \eqref{eq:sparsity} \& Theorem \ref{thm:LLR} we deduce that
\begin{align*}
\maxflow(\Hs[n,d]; \mu) &\geq \Omega\left(\frac{n}{\dmax \cdot \log m}\right).
\end{align*}
We can therefore route $\mu$ on $\Hs[n,d]$ with maximum edge congestion at most $O\left(\frac{\dmax \cdot \log m}{n}\right)$.

Suppose that we assign unit capacity to every vertex of $\Hs[n,d]$.
Since every vertex of $\Hs[n,d]$ has degree $O(d\cdot n)$, it follows
that $\mu$ can be routed with maximum vertex congestion at most at
most $O\left(d \cdot \dmax \cdot \log m\right)$.
Arguing as in the previous proof, by applying the rounding techniques of Raghavan \& Thompson \cite{DBLP:journals/combinatorica/RaghavanT87}, and Raghavan \cite{DBLP:journals/jcss/Raghavan88}, we obtain in deterministic polynomial time an embedding with depth $O\left(\frac{d\cdot \dmax \cdot \log^2 m}{\log\log m}\right)$.
\end{proof}

\ifabstract
\section{An exact algorithm for packing unit balls in {\large $\mathbb{R}^{\lowercase{d}}$}}
\label{sec:packing-alg}
\else
\section{An exact algorithm for packing unit balls in $\mathbb{R}^{d}$}
\label{sec:packing-alg}
\fi
We present an $n^{O(k^{1-1/d})}$ time algorithm for finding a pairwise
nonintersecting set $k$ unit balls in $d$-dimensional space (generalizing the result of Alber and Fiala~\cite{MR2070519} for $d=2$). As the technique (combining a sweeping argument, brute force, and dynamic programming) is fairly standard, we keep the discussion brief. Exactly the same argument works for finding $k$ pairwise nonintersecting $d$-dimensional unit cubes.

\begin{theorem}
Let $d\geq 2$ be a fixed constant.
There exists an algorithm that, given a set $X$ of unit $d$-dimensional balls in $\mathbb{R}^d$ and an integer $k\geq 0$, decides  in time $n^{O(k^{1-1/d})}$ whether there exist $k$ pairwise nonintersecting balls in $X$.
\end{theorem}

\begin{proof}
Let $s=k^{1/d}$.
For any $i\in \{0,\ldots,s-1\}$, and for any $j\in [d]$, let
$H(i,j) = \{(x_1,\ldots,x_d)\in \mathbb{R}^d : \text{ there exists } r\in \mathbb{Z} \text{ s.t.~} x_{j} = i + r\cdot s\}$.
Note that $H(i,j)$ is the union of parallel $(d-1)$-dimensional hyperplanes, with any two consecutive ones being at distance $s$.
For any $i\in \{0,\ldots,s-1\}$, let
$H(i) = \bigcup_{j=1}^d H(i,j)$.
Observe that $\mathbb{R}^d \setminus H(i)$ is the union of open hypercubes.

Let $A\subset \mathbb{R}^d$ be a unit $d$-dimensional ball.
By the union bound, we have
\ifabstract
\begin{align*}
\mathbf{Pr}_{i\in \{0,\ldots,s-1\}} &\left[ A \cap H(i) \neq \emptyset\right]\\
 &\leq \sum_{j=1}^d \mathbf{Pr}_{i\in \{0,\ldots,s-1\}}\left[ A\cap H(i,j) \neq \emptyset \right]\\
 &= O(d/s) = O(1/s),
\end{align*}
\else
\begin{align*}
\mathbf{Pr}_{i\in \{0,\ldots,s-1\}} \left[ A \cap H(i) \neq \emptyset\right] &\leq \sum_{j=1}^d \mathbf{Pr}_{i\in \{0,\ldots,s-1\}}\left[ A\cap H(i,j) \neq \emptyset \right] = O(d/s) = O(1/s),
\end{align*}
\fi
where $i\in \{0,\ldots,s-1\}$ is chosen uniformly at random.

Suppose that there exists a subset $X^* \subseteq X$ of pairwise nonintersecting balls with $|X^*|=k$.
By linearity of expectation, we obtain
\ifabstract
\begin{align*}
\mathbf{E}_{i\in \{0,\ldots,s-1\}} [|\{A\in & X^* : A\cap H(i) \neq \emptyset\}|]\\
 &\leq \sum_{A\in X^*} \mathbf{Pr}_{i\in \{0,\ldots,s-1\}} \left[ A\cap H(i) \neq \emptyset \right]\\
 &= O(k/s) = O(k^{1-1/d}).
\end{align*}
\else
\begin{align*}
\mathbf{E}_{i\in \{0,\ldots,s-1\}} \left[|\{A\in X^* : A\cap H(i) \neq \emptyset\}|\right] &\leq \sum_{A\in X^*} \mathbf{Pr}_{i\in \{0,\ldots,s-1\}} \left[ A\cap H(i) \neq \emptyset \right]\\
 &= O(k/s) = O(k^{1-1/d}).
\end{align*}
\fi
By averaging, there exists $i^*\in \{0,\ldots,s-1\}$, such that 
\[
|\{A\in X^* : A\cap H(i^*) \neq \emptyset\}| = O(k^{1-1/d}).
\]

The algorithm proceeds as follows.
We guess a value $i\in \{0,\ldots,s-1\}$, and we guess a subset $Y$ of at most $O(k^{1-1/d})$ balls in $X$ that intersect $H(i)$.
Let
\ifabstract
$X' = X \setminus \{A\in X : A \cap H(i)\neq \emptyset\}$.
\else
\[
X' = X \setminus \{A\in X : A \cap H(i)\neq \emptyset\}.
\]
\fi
Let ${\cal C}$ be the set of open hypercubes in $\mathbb{R}^d \setminus H(i)$.
We partition $X'$ into a collection of subsets $\{X'_C\}_{C\in {\cal C}}$, where every $X'_C$ contains all the balls that intersect the open hypercube $C\in {\cal C}$.
For each $C\in {\cal C}$, we define the subset $X''_C\subseteq X'_C$ containing all balls that do not intersect any of the balls in $Y$, i.e.
\ifabstract
$X''_C = \{A\in X'_C : A\cap \bigcup_{A'\in Y} A' = \emptyset\}$.
\else
\[
X''_C = \{A\in X'_C : A\cap \bigcup_{A'\in Y} A' = \emptyset\}.
\]
\fi

We now proceed to compute a maximum set of pairwise nonintersecting balls in each $X''_C$.
By translating, we may assume that $C=(0,s)^d$.
For each $j\in \{0,\ldots,s-1\}$, let
\ifabstract
$C(j) = C \cap \left([j,j+1] \times [0,s]^{d-1}\right)$.
\else
\[
C(j) = C \cap \left([j,j+1] \times [0,s]^{d-1}\right).
\]
\fi
Let $X''_{C,j}$ be the set of all balls intersecting $C(j)$ and 
let $Y''_{C,j}=\bigcup_{j'=0}^{j'=j}X''_{C,j'}$.
For any $j\in \{0,\ldots,s-1\}$, there can be at most $O(s)$ balls in
the solution $X^*$ that intersect $C(j)$.  For each $j\in
\{0,\ldots,s-1\}$, we compute all possible subsets of at most $O(s)$
pairwise nonintersecting balls in $X''_{C,j}$; let $\X_{C,j}$ be the collection of all these sets.  We
can now compute an maximum set $Y_C$ of pairwise nonintersecting balls
in $X''_C$ via dynamic programming, as follows. For every $j\in
\{0,\ldots,s-1\}$ and subset $Z\in \X_{C,j}$, we compute the
maximum size of a set of pairwise nonintersecting balls in $Y''_{C,j}$
whose intersection with $X''_{C,j}$ is precisely $Z$. The important observation is that the sets $X''_{C,j}$ and $X''_{C,j-2}$ are disjoint. Hence the maximum for a given $j$ and $Z\in \X_{C,j}$ can be computed if we know the maximum for $j-1$ and every $Z\in \X_{C,j-1}$. 

After computing the maximum set $Y_C$ for each open hypercube $C$, we
output the set of pairwise nonintersecting balls in $X$ that we find
is $Y \cup \bigcup_{C\in {\cal C}} Y_C$.  The final set of balls is
the maximum such set computed for all choices of $i$, and $Y$.  This
concludes the description of the algorithm.

Let us first argue that the algorithm is correct.
Indeed, when we choose $i=i^*$, the algorithm will eventually correctly choose the correct set $Y=X^* \cap H(i^*)$.
Once we remove the balls that intersect $H(i^*)$, and all the balls that intersect the balls in $Y$, the remaining subproblems are independent, and are solved optimally.  Therefore, the resulting global solution is optimal.

Lastly, let us bound the running time.
There are $s=n^{O(1/d)}$ choices for $i$, and for each such choice, there are at most $n^{O(k^{1-1/d})}$ choices for $Y$.
For every such choice, we solve at most $n$ different subproblems (one for every open hypercube in ${\cal C}$).
Each subproblem uses dynamic programming with a table with $O(s)$ entries, where each entry stores $n^{O(s)}= n^{O(k^{1-1/d})}$ different partial solutions.
It follows that the total running time is $n^{O(k^{1-1/d})}$, as required.
\end{proof}

%%% Local Variables: 
%%% mode: latex
%%% TeX-master: "main"
%%% End: 

\end{document}